\documentclass[oneside,reqno]{amsart}
\usepackage[utf8]{inputenc}
\usepackage{amsmath, amsthm,amssymb,amsfonts,latexsym}
\usepackage[dvips]{graphicx}
\usepackage{color}
\usepackage{array}
\usepackage[hmargin=2.75cm,vmargin=2.75cm]{geometry}
\usepackage{booktabs, multicol, multirow}
\usepackage{verbatim}
\usepackage{natbib}
\newtheorem{prop}{Proposition}
\newtheorem{cor}{Corollary}
\newtheorem{lemma}{Lemma}
\newtheorem{thm}{Theorem}

\newtheorem{assump}{Assumption}
\newtheorem{defin}{Definition}
\usepackage{tikz}
\usepackage{caption}
\usepackage{subcaption}
\usetikzlibrary{trees,shapes,snakes}
\usepackage{enumitem}
\usepackage{bbm}

\theoremstyle{remark}
\newtheorem{remark}{Remark}

\allowdisplaybreaks

\newcommand{\ES}{\mathrm{ES}}
\newcommand{\R}{\mathbb{R}}

\newcommand{\N}{\mathbb{N}}

\title[The importance of dynamic risk constraints for limited liability operators]{
The importance of dynamic risk constraints for limited liability operators
} 

\author{John Armstrong}
\address{John Armstrong: King's College London}
\email{john.1.armstrong@kcl.ac.uk}

\author{Damiano Brigo}
\address{Damiano Brigo: Imperial College London}
\email{damiano.brigo@imperial.ac.uk}

\author{Alex S.L. Tse}
\address{Alex S.L. Tse: University of Exeter}
\email{s.tse@exeter.ac.uk}

\date{\today}
\begin{document}
\maketitle

\begin{abstract}
Previous literature shows that prevalent risk measures such as Value at Risk or Expected Shortfall are ineffective to curb excessive risk-taking by a tail-risk-seeking trader with S-shaped utility function in the context of portfolio optimisation. However, these conclusions hold only when the constraints are static in the sense that the risk measure is just applied to the terminal portfolio value. In this paper, we consider a portfolio optimisation problem featuring S-shaped utility and a dynamic risk constraint which is imposed throughout the entire trading horizon. Provided that the risk control policy is sufficiently strict relative to the asset performance, the trader's portfolio strategies and the resulting maximal expected utility can be effectively constrained by a dynamic risk measure. Finally, we argue that dynamic risk constraints might still be ineffective if the trader has access to a derivatives market. 
\end{abstract}

\section{Introduction}
\label{sect:intro}

Portfolio optimisations are typically formulated as an expected utility maximisation problem faced by a risk averse agent with concave utility function. However, a simple concave function may not be sufficient to model agents’ preferences in an actual trading environment. For example, the limited-liability feature of a financial institution as well as standard remuneration scheme tend to create incentive distortion where a successful trader can share the profits via bonuses but a failed trader can simply walk away without punishment. Thus gains and losses can be perceived very differently by an agent leading to deviation from a concave utility function. See for example \cite{carpenter00} and \cite{bichuch-sturm14}. At a psychological level, the seminal work of \cite{kahneman-tversky79} and many of the other follow-up studies reveal that individuals are risk averse over positive outcomes but risk seeking over negative outcomes. These stylised preferences can be better captured by an S-shaped utility function which is concave on gains and convex on losses. 

This paper concerns the risk-taking behaviours of “tail-risk-seeking traders” who do not care much about extreme losses and hence their utility function is  S-shaped. It is of great regulatory interests to understand how the trading activities of a tail-risk-seeking trader can be controlled by standard risk measures. A surprising result has been reported in a recent paper of \cite{armstrong-brigo19} that Value at Risk (VaR) and Expected Shortfall (ES) are totally ineffective to curb the risk-taking behaviours of tail-risk-seeking traders. They consider a portfolio optimisation problem under S-shaped utility function and find that the value function of the trader remains the same upon imposing a static VaR/ES constraint on the terminal portfolio value. In other words, neither VaR nor ES can alter the maximal expected utility attained by a tail-risk-seeking trader compared to the benchmark case without any risk constraint. This casts doubt over the usefulness of prevalent risk management protocols to combat excessive risk-taking by traders with more realistic preferences. An earlier restricted version of the same result, focusing only on a Black-Scholes option market, is in \cite{armstrong-brigo18}. A further related result in \cite{armstrong-brigo20} introduces the notion of $\rho$-arbitrage for a coherent risk measure $\rho$. Positive homogeneity of the measure $\rho$ is the key property that is used to reach the result. A risk measure $\rho$ is defined to be  ineffective if a static risk constraint based on that measure cannot lower the expected utility of a limited liability trader. A $\rho$--arbitrage is defined as a portfolio payoff with non-positive price, non-positive risk as measured by $\rho$ but with strictly positive probability of being strictly positive. The ineffectiveness of static risk constraints based on the coherent risk measure $\rho$ is shown to be equivalent to the existence of a $\rho$--arbitrage. Again, the emphasis for us, in this paper, is that also in \cite{armstrong-brigo18} and \cite{armstrong-brigo20} the risk constraints are static and that the situation becomes very different with dynamic risk constraints.

Indeed, in view of the above negative result, we explore a simple remedy which resurrects VaR/ES as a tool to risk manage tail-risk-seeking traders: the risk measure is imposed dynamically throughout the entire trading horizon. At each point of time given the current assets holding in place, the portfolio risk exposure is computed by projecting the distribution of the portfolio return over an evaluation window under the assumption that the assets holding remains unchanged. There are several advantages with such a dynamic risk constraint. First, this risk management approach is more consistent with the industrial practice where the risk exposure of the trader’s positions is typically reported and monitored at least daily. Second, imposing a static risk constraint on the terminal portfolio value only usually leads to a time-inconsistent optimisation problem where the optimal strategy solved at a future time point may not be consistent with the one derived in the past. This results in difficulty with interpreting the notion of optimality, and one has to make further assumptions (such as whether the agent can pre-commit to the optimal strategy derived at time-zero) to pin down a unique prediction of the trader's action. The idea of time-inconsistency in dynamic optimisation problems can be dated back to \cite{strotz55}. 

Our main contribution is to show that a dynamic VaR or ES constraint can indeed constrain a tail-risk-seeking trader, in the sense that the maximal expected utility attained can be reduced provided that the risk control policy is sufficiently strict relative to the Sharpe ratio of the risky asset. The difference between a static and a dynamic risk constraint is drastic both mathematically and economically. In a complete market, any arbitrary payoff can be synthesised by dynamic replication. As a result, the problem of solving for the optimal trading strategy is equivalent to finding a utility-maximising payoff whose no-arbitrage price is equal to the initial wealth available. This duality principle which converts a dynamic stochastic control problem into a static optimisation has been widely adopted to solve portfolio optimisation problems. A static risk measure applied to the terminal portfolio value only restricts the class of the admissible payoffs. \cite{armstrong-brigo19} show that one can construct a sequence of digital options which pay a small positive amount most of the time but incur an extreme loss with a tiny probability, and that these payoffs can be carefully engineered to satisfy any given VaR/ES limit. The resulting expected utilities will converge to the same utility level associated with an unconstrained problem. 

The conclusion changes significantly when the risk constraint is applied dynamically instead. To comply with the given risk limit at each time point, the notional invested in the underlying assets has to be capped if the risk policy is sufficiently strict. Thus a dynamic risk constraint now has first-order impact on the admissible trading strategies. The usual duality approach no longer works because the restriction on the trading strategies from the outset precludes dynamic replication of a claim. We therefore have to resort to the primal HJB equation approach to solve the portfolio optimisation problem. Although a close-form solution is not available in general, we can nonetheless deduce the analytical conditions on the model parameters under which a dynamic VaR/ES constraint becomes effective. In a special case where the excess return of the asset is zero, we can provide a finer characterisation of the optimal trading strategy. 

Our results show that a dynamic risk constraint can be effective against a ``delta-one trader'' who can only invest in the underlying risky assets. What will happen if the trader can access derivatives trading as well? In the context of utility maximisation under market completeness, there is no economic difference between delta-one and derivatives trading since any payoff can be replicated by dynamic trading in the underlying assets. We argue, however, that a dynamic risk constraint such as ES will become ineffective again if derivatives trading is allowed. The key idea is that a derivatives trader can exploit dynamic rebalancing to continuously roll-over some risky digital options to ensure the risk constraint is satisfied at all time while generating an arbitrarily high level of utility.

We conclude the introduction by discussing some related work. A vast literature on continuous-time portfolio optimisation has emerged since Merton (\citeyear{Merton:69}, \citeyear{Merton:71}). One natural extension of the original Merton model is to incorporate additional constraints in form of a risk functional applied to the terminal portfolio value. Examples of the extra constraints include VaR (\cite{basak-shapiro01}), expected loss and other similar shortfall-style measures (\cite{gabih-grecksch-wunderlich05}), probability of outperforming a given benchmark (\cite{boyle-tian07}) and utility-based shortfall risk (\cite{gundel-weber08}). In these papers, the combination of concave utility function and static risk constraint facilitates the use of the dual approach to solve the underlying optimisation problems. 

There has been a recent strand of literature focusing on dynamic risk constraints. \cite{yiu04}, \cite{cuoco-he-isaenko08}, \cite{akume-luderec-wunderlich10} consider similar portfolio optimisation problems with VaR/ES constraints under different modelling setups. The optimal trading strategy behaves very differently when a static constraint is replaced by a dynamic one. For example, \cite{basak-shapiro01} show that a static VaR constraint may induce the trader to take more risk (relative to the unconstrained case) in the bad state of the world, whereas \cite{cuoco-he-isaenko08} show that if the VaR constraint is applied dynamically then the optimal risk exposure can be unanimously reduced. HJB equation formulation has to be used when solving the problem with dynamic risk constraints. All the papers cited above work with a concave utility function, and thus the problem is still relatively standard to yield analytical and numerical progress. 

S-shaped utility maximisation has received a lot of attention in the context of behavioural economics and convex incentive scheme. Despite the non-standard shape of the underlying utility function, duality method can still be suitably adapted to solve the optimisation problems. See for example \cite{berkelaar-kouwenberg-post04}, \cite{reichlin2013}, \cite{bichuch-sturm14} and the references therein. Papers on dynamic portfolio optimisation which simultaneously feature S-shaped utility as well as VaR/ES constraint include \cite{armstrong-brigo19}, \cite{guan-liang16} and \cite{dong-zheng20}. But again, the constraints are static in nature which are only imposed at the terminal time point. Our work fills the gap in the literature by considering S-shaped utility function and dynamic risk constraint in conjunction. In the same spirit that \cite{cuoco-he-isaenko08} is the dynamic version of \cite{basak-shapiro01} under concave utility function, our work can be viewed as the dynamic version of \cite{armstrong-brigo19} under S-shaped utility to give insights on the new economic phenomena when a more realistic risk management approach is adopted. 

The rest of the paper is organised as follows. Section \ref{sect:setup} gives an overview of the modelling framework. The main results of the paper are stated in Section \ref{sect:mainresults} with some numerical illustrations. A special case that the excess return of the asset being zero is analysed in details in Section \ref{sect:zerodrift}. We briefly discuss in Section \ref{sect:derivatives} how the results will change if the trader can access a derivatives market. Section \ref{sect:conclude} concludes. Miscellaneous technical materials are deferred to the appendix.

\section{Modelling setup}
\label{sect:setup}
\subsection{The economy}

For simplicity of exposition, in the main body of this paper we consider a standard Black-Scholes economy with a riskfree bond and one risky asset only. Extension to the multi-asset setup is discussed in the Appendix \ref{app:multi}.

Fix a terminal horizon $T>0$. Let $(\Omega,\mathcal{F},\{\mathcal{F}_t\}_{0\leq t\leq T},\mathbb{P})$ be a filtered probability space satisfying the usual conditions which supports a one-dimensional Brownian motion $B=(B_t)_{t\geq 0}$. The risky asset has price process $S=(S_t)_{t\geq 0}$ following a geometric Brownian motion
\begin{align*}
\frac{dS_t}{S_t}=\mu dt + \sigma dB_t
\end{align*}
with drift $\mu$ and volatility $\sigma>0$, and the riskfree bond has a constant interest rate of $r$. A trader invests in the two assets dynamically where an amount of $\Pi_t$ is invested in the risky asset at time $t$. The portfolio strategy $\Pi=(\Pi_t)_{t\geq 0}$ is said to be admissible if it is adapted and $\int_0^T \Pi_t^2 dt<\infty$ almost surely. The set of admissible portfolio strategies is denoted by $\mathcal{A}_0$. The portfolio value process $X=(X_t)_{t\geq 0}$ then evolves as
\begin{align}
dX_t = \frac{\Pi_t}{S_t} dS_t + r(X_t - \Pi_t)dt =rX_t dt+\Pi_t [(\mu-r) dt + \sigma dB_t],\qquad X_0=x_0,
\label{eq:wealth}
\end{align}
where $x_0$ is an exogenously given initial capital of the trader.

\subsection{Dynamic risk constraints}

Suppose $r\neq 0$ for the moment. The dynamics \eqref{eq:wealth} can be rewritten as
\begin{align*}
dX_t=\kappa(\theta_t-X_t) dt + \sigma \Pi_t dB_t
\end{align*}
with $\kappa:=-r$ and $\theta_t:=-\frac{(\mu-r)\Pi_t}{r}$. This is an Ornstein–Uhlenbeck process and thus
\begin{align*}
X_{t+\Delta}&=e^{-\kappa \Delta}X_t + \kappa\int_t ^ {t+\Delta} e^{-\kappa(t+\Delta - s)}\theta_s ds + \sigma\int_t ^ {t+\Delta} e^{-\kappa(t+\Delta - s)}\Pi_s dB_s
\end{align*}
for any $t$ and $\Delta>0$. We then deduce
\begin{align}
X_{t+\Delta}-e^{r\Delta}X_t= (\mu-r)\int_t ^ {t+\Delta} e^{r(t+\Delta - s)}\Pi_s ds +\sigma \int_t ^ {t+\Delta} e^{r(t+\Delta - s)}\Pi_s dB_s
\label{eq:portreturn}
\end{align}
which could be interpreted as the (numeraire-adjusted) portfolio gain/loss over the time horizon $[t,t+\Delta]$. 

At each instant of time $t$, a risk manager assesses the risk associated with the portfolio return given by \eqref{eq:portreturn}. Since the risk manager typically does not have the knowledge of the trader's portfolio strategy beyond the current time $t$, he assumes the portfolio strategy $\Pi$ will be held fixed over the risk evaluation window $[t,t+\Delta]$. Then the time-$t$ estimated random variable of portfolio loss over $[t,t+\Delta]$, denoted by $L_t$, is given by
\begin{align*}
-L_t&:= (\mu-r)\Pi_t\int_t ^ {t+\Delta} e^{r(t+\Delta - s)} ds +  \sigma\Pi_t \int_t ^ {t+\Delta} e^{r(t+\Delta - s)} dB_s \\
&=\frac{(\mu-r)(e^{r\Delta}-1)}{r}\Pi_t+\sigma \Pi_t \int_t ^ {t+\Delta} e^{r(t+\Delta - s)} dB_s
\end{align*}
and as such $L_t$ is normally distributed with mean and variance of
\begin{align}
\mathbb{E}[L_t]=-\frac{(\mu-r)(e^{r\Delta}-1)}{r}\Pi_t,\quad \text{Var}(L_t)=\frac{(e^{2r\Delta}-1)\sigma^2}{2r}\Pi_t^2.
\label{eq:lossdist}
\end{align}
The special case of $r=0$ can be recovered by considering the appropriate limits in \eqref{eq:lossdist}. 

\begin{remark}
In the literature, there are multiple ways to estimate the projected distribution of portfolio gain/loss. Our approach is based on \cite{yiu04} where the notional invested in the risky asset $\Pi_t$ is assumed to be fixed by the risk manager. Alternatively, the risk manager can also assume the proportion of capital invested in the risky asset $\Pi_t/X_t$ is fixed - this assumption is adopted for example by \cite{cuoco-he-isaenko08}. The latter approach leads to a more difficult mathematical problem in general because the projected distribution will then also depend on the current portfolio value $X_t$. The question about which approach is more superior depends on the risk management practice adopted at a particular institution. Another very plausible approach is to assume the quantity of the assets $n_t:=\Pi_t/S_t$ to be fixed (this could be more relevant in the context of equity trading where stock and future positions are typically recorded in terms of quantity rather than notional). Then starting from \eqref{eq:portreturn} we can deduce that
\begin{align*}
-L_t&:= (\mu-r)n_t\int_t ^ {t+\Delta} e^{r(t+\Delta - s)} S_u du +  \sigma n_t \int_t ^ {t+\Delta} e^{r(t+\Delta - u)} S_u dB_u \\
&= n_t\int_t^{t+\Delta}e^{r(t+\Delta-u)} \left((\mu-r)S_u du+\sigma S_u dB_u\right)\\
&=n_t e^{r(t+\Delta)} \int_t^{t+\Delta} e^{-ru}(dS_u-rS_u du)\\
&=n_t e^{r(t+\Delta)} \int_t^{t+\Delta} d(e^{-ru}S_u)\\
&=n_t(S_{t+\Delta}-e^{r\Delta}S_t)\\
&=\Pi_t\left[e^{\left(\mu-\frac{\sigma^2}{2}\right)(t+\Delta)+\sigma(B_{t+\Delta}-B_t)}-e^{r\Delta}\right]
\end{align*}
which only depends on the current state via $\Pi_t=n_t S_t$. This is qualitatively very similar to the approach used by \cite{yiu04} and us, except that $L_t$ is now linked to some log-normal random variable.
\end{remark}

A dynamic risk constraint is imposed such that $\rho(L_t)\leq R$ for all $t\in[0,T)$. Here $\rho(\cdot)$ is some risk measure and $R>0$ is an exogenously given level of risk limit. For example, if the risk measure is taken as VaR with confidence level $\alpha$ (with $\alpha<0.5$) such that $\rho(L_t)=\text{VaR}_{\alpha}(L_t):=\sup\{x\in\mathbb{R}:\mathbb{P}(L_t\geq x)> \alpha\}$, then using the Gaussian property of $L_t$ and \eqref{eq:lossdist} the constraint can be specialised to
\begin{align*}
-\frac{(\mu-r)(e^{r\Delta}-1)}{r}\Pi_t - \sigma \sqrt{\frac{e^{2r\Delta}-1}{2r}}\Phi^{-1}(\alpha) |\Pi_t|\leq R
\end{align*}
where $\Phi$ denotes the cumulative distribution function (cdf) of a $N(0,1)$ random variable. We define the set
\begin{align}
K_{\text{VaR}}:=\left\{\pi\in\mathbb{R}:-\frac{(\mu-r)(e^{r\Delta}-1)}{r}\pi - \sigma \sqrt{\frac{e^{2r\Delta}-1}{2r}}\Phi^{-1}(\alpha) |\pi|\leq R\right\}
\label{eq:setKVaR}
\end{align}
such that compliance with the dynamic VaR constraint at time $t$ is equivalent to $\Pi_t\in K_{\text{VaR}}$.

Similarly, if the risk measure is taken as ES with confidence level $\alpha$ such that $\rho(L_t)=\mathbb{E}[L_t| L_t\geq \text{VaR}_{\alpha}(L_t)]$, then the constraint becomes
\begin{align*}
-\frac{(\mu-r)(e^{r\Delta}-1)}{r}\Pi_t+\sigma \sqrt{\frac{e^{2r\Delta}-1}{2r}}\frac{\phi(\Phi^{-1}(\alpha))}{\alpha}|\Pi_t|\leq R
\end{align*}
with $\phi(\cdot)$ being the probability density function (pdf) of a $N(0,1)$ random variable. We then define the set
\begin{align}
K_{\text{ES}}:=\left\{\pi\in\mathbb{R}:-\frac{(\mu-r)(e^{r\Delta}-1)}{r}\pi+\sigma \sqrt{\frac{e^{2r\Delta}-1}{2r}}\frac{\phi(\Phi^{-1}(\alpha))}{\alpha}|\pi|\leq R\right\}
\label{eq:setKES}
\end{align}
where we require $\Pi_t\in K_{\text{ES}}$ for all $t$ in order to satisfy the dynamic ES constraint.

It turns out that the nature of the sets $K_{\text{VaR}}$ and $K_{\text{ES}}$ crucially depends on the Sharpe ratio of the risky asset $\frac{\mu-r}{\sigma}$, as the following lemma shows.

\begin{lemma}
Define the constants
\begin{align}
M_{\text{VaR}}:=-\sqrt{\frac{e^{2r\Delta}-1}{2r}}\frac{r}{e^{r\Delta}-1}\Phi^{-1}(\alpha)>0,\qquad
M_{\text{ES}}:=\sqrt{\frac{e^{2r\Delta}-1}{2r}}\frac{r}{e^{r\Delta}-1}\frac{\phi(\Phi^{-1}(\alpha))}{\alpha}>0.
\label{eq:Mconst}
\end{align}
Then for $i\in\{\text{VaR},\text{ES}\}$, the sets $K_i$ defined in \eqref{eq:setKVaR} and \eqref{eq:setKES} have the following properties:
\begin{enumerate}
	\item If $\frac{\mu-r}{\sigma}\geq M_{i}$, there exists $-\infty<k^i_1<0$ such that $K_{i}=[k^i_1,\infty)$;
	\item If $|\frac{\mu-r}{\sigma}|<M_{i}$, there exists $-\infty<k^i_1<0<k^i_2<\infty$ such that $K_{i}=[k^i_1,k^i_2]$;
	\item If $\frac{\mu-r}{\sigma}\leq -M_{i}$, there exists $0<k^i_2<\infty$ such that $K_{i}=(-\infty,k^i_2]$.
\end{enumerate}
Moreover,
\begin{align*}
k_1^\text{VaR}=-\frac{R}{ - \sigma \sqrt{\frac{e^{2r\Delta}-1}{2r}}\Phi^{-1}(\alpha)+\frac{(\mu-r)(e^{r\Delta}-1)}{r}},\qquad k_2^\text{VaR}=\frac{R}{ - \sigma \sqrt{\frac{e^{2r\Delta}-1}{2r}}\Phi^{-1}(\alpha)-\frac{(\mu-r)(e^{r\Delta}-1)}{r}};\\
k_1^\text{ES}=-\frac{R}{\sigma \sqrt{\frac{e^{2r\Delta}-1}{2r}}\frac{\phi(\Phi^{-1}(\alpha))}{\alpha}+\frac{(\mu-r)(e^{r\Delta}-1)}{r}},\qquad k_2^\text{ES}=\frac{R}{\sigma \sqrt{\frac{e^{2r\Delta}-1}{2r}}\frac{\phi(\Phi^{-1}(\alpha))}{\alpha}-\frac{(\mu-r)(e^{r\Delta}-1)}{r}}.
\end{align*}

\label{lem:setA}
\end{lemma}

\begin{proof}
	This is a simple exercise of analysing the piecewise linear function arising in the definition of $K_{\text{VaR}}$ and $K_{\text{ES}}$.
\end{proof}

The constants $M_{i}$ defined in \eqref{eq:Mconst} encapsulate the risk management parameters $\alpha$ and $\Delta$. Unless the quality of the investment asset is very good (measured by the magnitude of its Sharpe ratio) relative to $M_i$, a dynamic VaR or ES constraint will result in a restriction that $\Pi_t$ needs to take value in a bounded set, i.e. a delta limit restriction where both the long and short position in the underlying asset cannot exceed certain notional levels given by $k_1^{i}$ and $k_2^i$. It is also not hard to see that $M_i$ is decreasing in both $\alpha$ and $\Delta$. Hence a small confidence level of the VaR/ES constraint or a tight risk evaluation window will more likely lead to a bounded investment set $K_i$. Provided that $k^i_1$ and $k^i_2$ exist, one can also easily check that $|k^i_1|$ and $k^i_2$ are both decreasing in $\sigma$ and increasing in $R$ and $\alpha$. Hence a high asset volatility, low risk limit or tight confidence level of the VaR/ES measure will result in small absolute delta notional limit.

\subsection{Trader's utility function and optimisation problem}

We assume that the trading decision is made by a ``tail-risk-seeking trader'' who is insensitive towards extreme losses. His utility function $U(\cdot)$ is S-shaped and his goal is to maximise the expected utility of the terminal portfolio value. The only assumption required over $U$ is the following.
\begin{assump}
	The utility function $U:\mathbb{R}\to\mathbb{R}$ is a continuous, increasing and concave (resp. convex) function on $x>0$ (resp. $x<0$) with $U(0)=0$ and $\lim_{x\to -\infty}\frac{U(x)}{x}=0$.
\label{assump:standing}
\end{assump}
In particular, the trader is locally risk averse over the domain of gains but locally risk seeking over the domain of losses. Moreover, the assumption on the left-tail behaviour of the utility function further suggests that the trader is tail-risk-seeking in that the ``dis-utility'' due to extreme losses has a sub-linear growth. We do not require $U(x)$ to be differentiable. This allows us to consider for example the piecewise power utility function of \cite{kahneman-tversky79} which is not differentiable at $x=0$, or an option payoff function which may contain kinks.

Mathematically, the underlying optimisation problem is
\begin{align}
V(t,x):=\sup_{\Pi\in\mathcal{A}(K)}\mathbb{E}^{(t,x)}[U(X^{\Pi}_T)]
\label{eq:valFun}
\end{align}
where $X=X^{\Pi}$ has dynamics described by \eqref{eq:wealth}, and $\mathcal{A}(K)$ is the admissible set of the portfolio strategies under a given dynamic risk constraint in form of
\begin{align*}
\mathcal{A}(K):=\{\Pi\in \mathcal{A}_0:  \Pi(t,\omega)\in K\quad \mathcal{L}\otimes \mathbb{P}\text{-a.e. } (t,\omega)\}
\end{align*}
with $K\subseteq \mathbb{R}$ being some given set and $\mathcal{L}$ is Lebesgue measure. For example, if the risk constraint is absent we simply take $K=K_0:=\mathbb{R}$ and then $\mathcal{A}(K_0)=\mathcal{A}_0$. If a dynamic VaR constraint is in place, we set $K=K_{\text{VaR}}$ as defined in \eqref{eq:setKVaR}. Likewise a choice of $K=K_{\text{ES}}$ given by \eqref{eq:setKES} corresponds to a dynamic ES constraint.

\begin{remark}
Portfolio optimisation problem in form of \eqref{eq:valFun} with $U$ being a strictly concave, twice-differentiable function is studied by \cite{cvitanic-karatzas92}. Their results cannot be applied to our setup because our utility function is S-shaped. \cite{dong-zheng19} consider a version of the problem with S-shaped utility and short-selling restrictions. Their solution method is based a concavification argument in conjunction with the results by \cite{bian-miao-zheng11} which cover non-smooth utility function but only under the assumption that the set $K$ is in form of a convex cone. For our model, Lemma \ref{lem:setA} suggests that the set $K$ under VaR/ES constraint cannot be a convex cone. Thus we cannot apply their approaches to solve our problem.
\end{remark}

Let $V_i(t,x)$ be the value function of problem \eqref{eq:valFun} under $K=K_{i}$ with $i\in\{0,\text{VaR},\text{ES}\}$ denoting the label identifying which dynamic risk measure is being adopted (i.e no risk constraint at all, Value at Risk and Expected Shortfall). We first state a benchmark result based on \cite{armstrong-brigo19}.

\begin{prop}[Theorem 4.1 of \cite{armstrong-brigo19}]
The value function of the unconstrained portfolio optimisation problem is $V_{0}(t,x)=\sup_{s}U(s)$.
\label{prop:uncon}
\end{prop}

\begin{proof}[Sketch of proof]
Without loss of generality we just need to prove the result at $t=0$. By standard duality argument (see for example \cite{karatzas-lehoczky-shreve87}), the portfolio optimisation problem \eqref{eq:valFun} without any additional risk constraint is equivalent to solving
\begin{align*}
\begin{cases}
V_0(0,x_0)=\sup\limits_{X_T\in\mathcal{F}_T}\mathbb{E}[U(X_T)]\\
\mathbb{E}[\xi_T X_T]\leq x_0
\end{cases}
\end{align*}
where 
\begin{align*}
\xi_T:=\exp\left[\left(-\frac{\mu-r}{\sigma}\right)B_T-\left(r+\frac{1}{2}\left(\frac{\mu-r}{\sigma}\right)^2\right) T\right]
\end{align*}
is the pricing kernel in the Black-Scholes economy. Now consider a digital payoff in form of
\begin{align}
X_T=-\frac{b}{\mathbb{P}(\xi_T>k)} {\mathbbm 1}_{(\xi_T>k)}+b {\mathbbm 1}_{(\xi_T\leq k)} 
\label{eq:digital}
\end{align}
for $b>0$ and $k>0$. The budget constraint can be written as
\begin{align*}
&-\frac{b}{\mathbb{P}(\xi_T>k)}\mathbb{E}[\xi_T {\mathbbm 1}_{(\xi_T>k)}]+b\mathbb{E}[\xi_T {\mathbbm 1}_{(\xi_T\leq k)}]\leq x_0 \iff -b\leq \frac{\mathbb{P}(\xi_T>k)}{\mathbb{E}[\xi_T {\mathbbm 1}_{(\xi_T>k)}]}\left[x_0-b\mathbb{E}[\xi_T {\mathbbm 1}_{(\xi_T\leq k)}]\right].
\end{align*}
If $\xi_T$ is unbounded from the above (which is the case in the Black-Scholes model), then $\lim_{k\to\infty}\frac{\mathbb{P}(\xi_T>k)}{\mathbb{E}[\xi_T {\mathbbm 1}_{(\xi_T>k)}]}=0$. In turn for any $b>0$ fixed one can always find a sufficiently large $k$ such that the budget constraint is satisfied. The value function must be no less than the expected utility attained by this payoff structure, i.e.
\begin{align*}
V_0(0,x_0)\geq U\left(-\frac{b}{\mathbb{P}(\xi_T>k)}\right) \mathbb{P}(\xi_T>k)+U(b)\mathbb{P}(\xi_T<k).
\end{align*}
Under Assumption \ref{assump:standing}, $\lim_{x\to -\infty}\frac{U(x)}{x}=0$. Thus on sending $k\to\infty$ we deduce $V_0(0,x_0)\geq U(b)$. The result follows since $b>0$ is arbitrary.
\end{proof}

Without any risk constraint in place, the tail-risk-seeking trader can attain any arbitrarily high utility by replicating a sequence of digital options which pay a positive amount with a large probability but incur an extremely disastrous loss with very small probability. \cite{armstrong-brigo19} show that this result does not change even if a static VaR/ES constraint is imposed on the terminal portfolio value, in the sense that the trader can still manipulate the digital structure to attain an arbitrarily high utility level while satisfying the additional constraints.

We are interested in studying whether such conclusion will change if we adopt a dynamic risk constraint instead. With the unconstrained optimisation problem as our benchmark, we first give below a formal definition of the effectiveness of a dynamic risk constraint.

\begin{defin}
A dynamic risk constraint $i\in\{\text{VaR},\text{ES}\}$ is said to be effective if for each $t<T$ there exists $x$ such that $V_{i}(t,x)<V_{0}(t,x)=\sup_s U(s)$.
\label{def:effective}
\end{defin}

The notion of effectiveness in Definition \ref{def:effective} may appear to be somewhat weak as we do not insist that the trader's expected utility have to be strictly reduced at all states $(t,x)$. Indeed for a general utility function, we cannot expect $V_i(t,x)<\sup_s U(s)$ for all $(t,x)$. For example, consider a call spread payoff $U(x)=(x+1)^{+}-(x-1)^{+}-1$ which is a S-shaped, then for as long as $\Pi_t=0$ for all $t$ is an admissible strategy under a given dynamic risk constraint $i$ and interest rate is non-negative, we always have $V_i(t,x)=1=\sup_s U(s)$ for all $t<T$ and $x\geq 1$.  

\section{Main results}
\label{sect:mainresults}

We first give a useful proposition which is the building block of the main results in this paper.

\begin{prop}
For the optimisation problem \eqref{eq:valFun}, if the set $K$ is bounded then for every $t<T$ there exists $x$ such that $V(t,x)<\sup_{s}U(s)$.
\label{prop:delta_suff}
\end{prop}

\begin{proof}
	
	
	Since $U(x)\leq 0$ for $x\leq 0$ and $U(x)$ is concave on $x>0$, for any constant $C>0$ there always exists $m>0$ such that $U(x)\leq \bar{U}(x):=m x^{+}+C$ for all $x$. Then
	\begin{align*}
	J(t,x;\Pi):=\mathbb{E}^{(t,x)}[U(X^{\Pi}_T)]\leq \mathbb{E}^{(t,x)}[\bar{U}(X^{\Pi}_T)].
	\end{align*} 
	Since $K$ is bounded, there exits $b\in(0,\infty)$ such that $K\subseteq [-b,b]=:\bar{K}$. Then
	\begin{align}
	V(t,x)=\sup_{\Pi_t\in\mathcal{A}(K)}J(t,x;\Pi)&\leq \sup_{\Pi_t\in\mathcal{A}(K)}\mathbb{E}^{(t,x)}[\bar{U}(X^{\Pi}_T)]\nonumber \\
	&\leq \sup_{\Pi_t\in\mathcal{A}(\bar{K})}\mathbb{E}^{(t,x)}[\bar{U}(X^{\Pi}_T)]=:\bar{V}(t,x).
	\label{eq:upbound}
	\end{align}
	
	We can now derive the expression of $\bar{V}(t,x)$ as the value function of a stochastic control problem with payoff function $\bar{U}$ which is increasing and convex. Formally, we expect $\bar{V}(t,x)$ to be the (viscosity) solution of the HJB equation
	\begin{align}
	\begin{cases}
	-V_t-\sup\limits_{\pi\in [-b,b]}\left\{[(\mu-r)\pi+rx]V_x+\frac{\sigma^2}{2}V_{xx} \pi^2\right\}=0,& t<T;\\
	V(t,x)=\bar{U}(x),& t=T.
	\end{cases}
	\label{eq:hjbmain}
	\end{align}
	
	Suppose $\mu\geq r$ and recall that $\bar{U}$ is convex. Since the dynamics of the portfolio process is $dX_t= [rX_t+(\mu-r)\Pi_t]dt+\sigma \Pi_t dB_t$ where its drift and volatility are both increasing in $\Pi_t$, we expect the optimal strategy is to choose the largest possible value of $\Pi_t$ within the bounded set $\bar{K}$. Hence the candidate optimal control for problem \eqref{eq:upbound} is $\Pi^*_t=b<\infty$. The corresponding candidate value function is thus
	\begin{align*}
	 w(t,x):=\mathbb{E}^{(t,x)}[U_2(X^{\Pi^*}_T)]
	\end{align*}
	and the wealth process under the candidate optimal control is
	\begin{align*}
	dX^*_t=rX^*_tdt + b\left[(\mu-r)dt+\sigma  dB_t\right].
	\end{align*}
	Then
	\begin{align*}
	X^*_s=x e^{r(s-t)}+\frac{b(\mu-r)}{r}(e^{r(s-t)}-1)+b\sigma \int_t^s e^{r(s-u)}dB_u \qquad\text{for }s\geq t \text{ and } X^*_t=x,
	\end{align*}
	such that $X^{*}_T$ is normally distributed with mean $xe^{r(T-t)}+\frac{b(\mu-r)}{r}(e^{r(T-t)}-1)$ and variance $\frac{b^2\sigma^2}{2r}(e^{2r(T-t)}-1)$. Upon evaluating the expectation, we obtain
	\begin{align}
	w(t,x)&=\mathbb{E}[C+m(X_T^*)^{+}]\nonumber \\
	&= C + m \Bigl[xe^{r(T-t)}+\frac{b(\mu-r)}{r}(e^{r(T-t)}-1)\Bigl] \Phi\left(\frac{xe^{r(T-t)}+\frac{b(\mu-r)}{r}(e^{r(T-t)}-1)}{b\sigma\sqrt{\frac{e^{2r(T-t)}-1}{2r}}}\right) \nonumber \\
	&\qquad + m b\sigma\sqrt{\frac{e^{2r(T-t)}-1}{2r}}\phi\left(\frac{xe^{r(T-t)}+\frac{b(\mu-r)}{r}(e^{r(T-t)}-1)}{b\sigma\sqrt{\frac{e^{2r(T-t)}-1}{2r}}}\right)
	\label{eq:upbound2}
	\end{align}
	where $\Phi$ and $\phi$ are the cdf and pdf of a standard $N(0,1)$ random variable respectively. $w(t,x)$ is indeed $C^{1\times 2}$ on $[0,T)\times \mathbb{R}$, and is increasing convex in $x$. It can be easily shown that $w$ is a solution to the HJB equation \eqref{eq:hjbmain}. Standard verification arguments then lead to the conclusion that $\bar{V}(t,x)=w(t,x)$. Finally, for each fixed $t$ we have $\bar{V}(t,x)=w(t,x)\to C$ as $x\downarrow -\infty$. But the constant $C>0$ can be arbitrarily chosen. Using the fact that $V(t,x)\leq \bar{V}(t,x)$, the desired result follows if we choose $C\in(0,\sup_s U(s))$. The case of $\mu<r$ can be handled similarly except that the optimal control will become $\Pi^*_t=-b$ instead.
\end{proof}

The implication of Proposition \ref{prop:delta_suff} is that a delta notional limit on the risky asset alone is sufficient to constrain a tail-risk-seeking trader. For an unconstrained problem, as discussed in the proof of Proposition \ref{prop:uncon} one can attain an arbitrarily high utility level by replicating some digital options. But it is known that the delta of a digital option can be unboundedly large when the time to maturity becomes short and the underlying stock price is near the strike. Hence a trader cannot replicate a digital option and hold the position until maturity while complying the dynamic risk constraint with certainty. In practice, a trading desk with a substantial at-the-money digital option position with short maturity will often be requested to wind-down the trade to reduce the pin risk.

Next we state the main theorem of this paper which provides a precise condition under which a dynamic VaR/ES constraint can effectively restrict a rough trader.

\begin{thm}
Recall the constants $M_{\text{VaR}}$ and $M_{\text{ES}}$ introduced in \eqref{eq:Mconst}. A dynamic Value at Risk constraint is effective if and only if $|\frac{\mu-r}{\sigma}|<M_{\text{VaR}}$. A dynamic Expected Shortfall constraint is effective if and only if $|\frac{\mu-r}{\sigma}|<M_{\text{ES}}$.
\label{thm:effective}
\end{thm}

\begin{proof}
	In view of Lemma \ref{lem:setA} and Proposition \ref{prop:delta_suff} we only need to prove the ``only if'' part of the theorem. 

	In the proof of Proposition \ref{prop:uncon}, a utility level of $\sup_{s}U(s)$ can be attained by replicating a sequence of payoffs in form of $$X_T=a {\mathbbm 1}_{(\xi_T>k)}+b {\mathbbm 1}_{(\xi_T\leq k)}$$ with $a<0<b$ where $\xi_T$ is the pricing kernel in the Black-Scholes economy. But $$\xi_T=\exp\left[-\frac{\mu-r}{\sigma} B_T-\left(r+\frac{1}{2}\lambda^2\right) T\right]\propto S_T^{-\frac{\mu-r}{\sigma^2}}.$$ Hence $X_T$ is increasing (resp. decreasing) in $S_T$ if $\mu> r$ (resp. $\mu< r$). 
	
	Suppose $\frac{\mu-r}{\sigma}\geq M_{i}>0$ where $i\in\{\text{VaR},\text{ES}\}$. Then by Lemma \ref{lem:setA} the admissible set is in form of $[k_1^i,\infty)$ where $k_1^i\in(-\infty,0)$. In other words, there is no restriction on the investment level for as long as only long position is taken. But since $\mu>r$, if we view $X_T$ as a contingent claim written on the risky asset, the payoff $X_T=X(S_T)$ is an increasing function and thus the option must have non-negative delta for all $(t,x)$. Hence only long position is ever required to replicate this claim. The sequence of strategies replicating the digital options which yield a utility level of $\sup_s U(s)$ must also belong to $\mathcal{A}(K_i)$ as well. In this case, the dynamic risk $i$ constraint is not effective. Similar results hold for the case of $\frac{\mu-r}{\sigma}\leq -M_i<0$.
\end{proof}

A dynamic risk constraint $i\in\{\text{VaR},\text{ES}\}$ restricts a tail-risk-seeking trader if and only if the (magnitude of) Sharpe ratio is smaller than the constant $M_i$. Surprisingly, from the definition of $M_i$ in \eqref{eq:Mconst} we see that it does not depend on the risk limit level $R$ at all but only the evaluation horizon $\Delta$, confidence level $\alpha$ and interest rate $r$. In other words, increasing the risk limit alone is not sufficient to guarantee the effectiveness of a dynamic risk measure. The risk manager must impose a short evaluation horizon window (small $\Delta$) and emphasise on the extreme tail of the loss distribution (small $\alpha$) to ensure the necessary and sufficient condition of dynamic risk measure effectiveness $|\frac{\mu-r}{\sigma}|<M_{i}$ is satisfied. But given a dynamic risk constraint is effective, the risk limit $R$ will play a role in controlling the implied delta notional limit as per the expressions of $k^i_1$ and $k^i_2$ in Lemma \ref{lem:setA}.

The main driver behind the effectiveness of a dynamic risk measure is that the risk constraint implies a hard bound on the delta notional to be taken by the trader. Indeed, there is no economic difference between imposing a delta limit and a more complicated risk measure such as VaR or ES, as the following corollary shows.
\begin{cor}
An effective dynamic risk constraint $i\in\{\text{VaR},\text{ES}\}$ is equivalent to imposing a delta notional limit on the underlying risky asset. i.e. if a dynamic constraint $i$ is effective, then there exists a bounded set $D\subseteq \mathbb{R}$ such that
\begin{align*}
V_i(t,x)=V_d(t,x):=\sup_{\Pi\in\mathcal{A}(D)}\mathbb{E}^{(t,x)}[U(X_T)].
\end{align*} 
\label{cor:equi}
\end{cor}

\begin{proof}
This follows immediately from Lemma \ref{lem:setA}.
\end{proof}

The next proposition gives a theoretical characterisation of the value function. 

\begin{prop}
Suppose the model parameters are such that  $|\frac{\mu-r}{\sigma}|<M_i$. Then the value function of the optimisation problem \eqref{eq:valFun} under dynamic risk constraint $i$ is the unique viscosity solution to the HJB equation
\begin{align}
-V_t-H_i(x,V_x,V_{xx})=0,\qquad t<T,
\label{eq:hjb}
\end{align}
subject to terminal condition $V(T,x)=U(x)$ and linear growth condition $V(t,x)\leq c(1+|x|)$ for some $c>0$. Here $H_i$ is the Hamiltonian defined as
\begin{align}
H_i(x,p,M):=\sup_{\pi\in K_i}\left\{[(\mu-r)\pi+rx]p+\frac{\sigma^2}{2}M \pi^2\right\}.
\label{eq:ham}
\end{align}
\label{prop:vis}
\end{prop}

\begin{proof}
Provided that $|\frac{\mu-r}{\sigma}|<M_i$, the set $K_i$ is bounded and hence by \eqref{eq:upbound2} we can deduce that $V(t,x)\leq \alpha_1+\beta_1 x$ for some constant $\alpha_1>0$ and $\beta_1>0$. On the other hand, the utility function $U$ is a negative convex increasing function on $x<0$. Hence there exists $\alpha_2>0$ and $\beta_2>0$ such that $U(x)>-\alpha_2 -\beta_2 x^{-}=:G(x)$ for all $x$. Then since $\tilde{\Pi}_t=0$ for all $t$ is an admissible strategy in $\mathcal{A}(K_i)$, we have
\begin{align*}
V_i(t,x)=\sup_{\Pi\in\mathcal{A}(K_i)}\mathbb{E}^{(t,x)}\left[U(X_T^\Pi)\right]\geq \mathbb{E}^{(t,x)}\left[G(X_T^{\tilde{\Pi}})\right]=-\alpha_2-\beta_2 e^{r(T-t)}x^{-}
\end{align*}
for all $(t,x)$. Thus we conclude $V_i(t,x)\leq c(1+|x|)$ for some $c>0$, i.e. the value function has at most a linear growth. Finally, since $K_i$ is bounded the Hamiltonian in \eqref{eq:ham} is always finite. Standard theory of stochastic control suggests that the value function $V_i$ is a viscosity solution to the HJB equation \eqref{eq:hjb}. Moreover, the solution is indeed unique in the class of viscosity solutions with linear growth due to strong comparison principle. See Theorem 4.4.5 of \cite{pham09}.
\end{proof}

Proposition \ref{prop:vis} provides a characterisation of the value function in terms of viscosity solution, which serves as a useful basis for implementation of numerical methods to solve the HJB equation. In general, it is difficult to make further analytical progress to extract meaningful economic intuitions from the solution structure. Nonetheless, in Section \ref{sect:zerodrift} we will show that further characterisation of the optimal portfolio strategy is indeed possible under a special case of $\mu=r$. 

For now, we numerically solve the portfolio optimisation problem for the more general case of $\mu\neq r$. Two specifications of utility function are considered: the \cite{kahneman-tversky79} piecewise power form of
\begin{align*}
U(x)=
\begin{cases}
x^{\beta_1},&x\geq 0 \\
-k|x|^{\beta_2},& x<0
\end{cases}
\end{align*}
with $0<\beta_1,\beta_2<1$ and $k>0$, and the piecewise exponential form of
\begin{align*}
U(x)=
\begin{cases}
\phi_1 (1-e^{-\gamma_1 x}),&x\geq 0 \\
\phi_2 (e^{\gamma_2 x}-1),& x<0
\end{cases}
\end{align*}
with $\phi_1,\phi_2,\gamma_1,\gamma_2>0$.
 
A fully implicit discretisation scheme with Newton-type policy iteration is used to solve the HJB equation \eqref{eq:hjb}. See \cite{forsyth-labahn07} for a description of the algorithm and the relevant conditions for convergence. The implementation of numerical methods is quite straightforward and we briefly discuss two practical issues relevant to our specific problem: First, the value function of our portfolio optimisation problem is defined on an unbounded domain $[0,T]\times \mathbb{R}$. As an approximation, we only solve for the numerical solutions on a bounded domain $[0,T)\times [-x_{min},x_{max}]$ for some large $x_{min}>0$ and $x_{max}>0$. An artificial boundary condition $V(t,x)=U(x)$ is imposed along $[0,T)\times\{-x_{min},x_{max}\}$ and then we focus on the solution behaviours on a narrow range away from the boundary points. We observe that the numerical results are not sensitive to the choice of $x_{min}$ and $x_{max}$ provided that their values are sufficiently large. Second, we focus on a parameter choice of $r=0$ to ensure that the ``positive coefficient condition'' of the finite difference scheme (Condition 4.1 of \cite{forsyth-labahn07}) is satisfied when the step size along the $x$-axis is sufficiently small. But the more general case of non-zero interest rate can be recovered by change of numeraire.

\begin{figure}[!htbp]
	\captionsetup[subfigure]{width=0.5\textwidth}
	\centering
	\subcaptionbox{S-shaped power utility function.}{\includegraphics[scale =0.515] {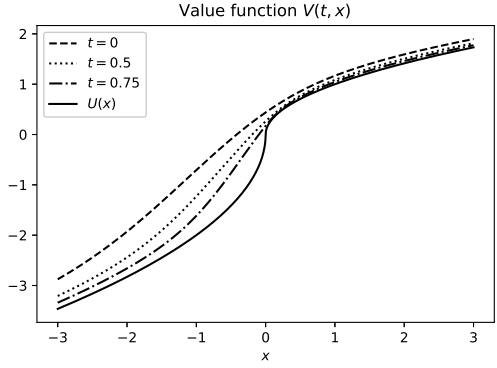}}
	\subcaptionbox{S-shaped power utility function.}{\includegraphics[scale =0.515]{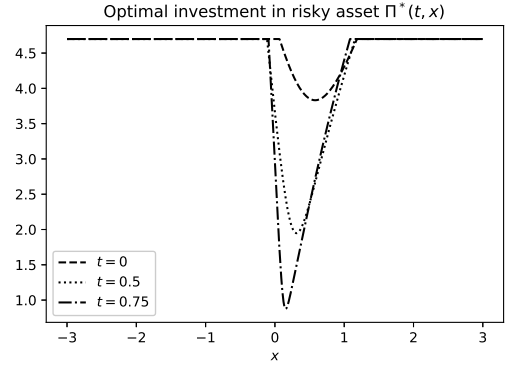}}
	\subcaptionbox{S-shaped exponential utility function.}{\includegraphics[scale =0.515]{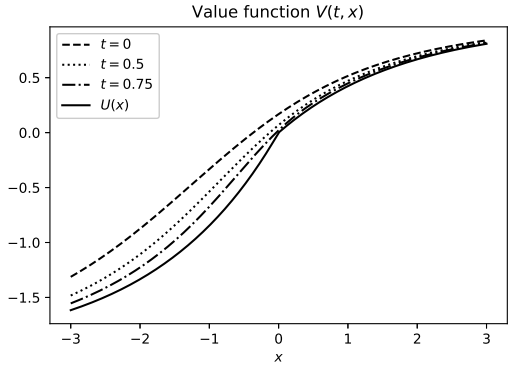}}
	\subcaptionbox{S-shaped exponential utility function.}{\includegraphics[scale =0.515]{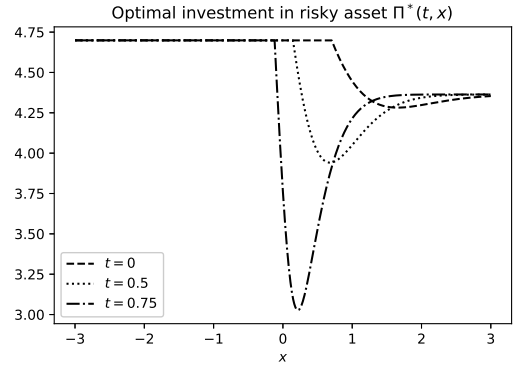}}
	\caption{Value function and optimal investment level at different selected time points under S-shaped power and exponential utility function. Parameters used are $\mu=0.15$, $\sigma=0.25$, $r=0$, $T=1$, $\alpha=0.01$, $R=1$, $\Delta=30/250$, $\beta_1=\beta_2=0.5$, $k=2$, $\gamma_1=\gamma_2=0.55$, $\phi_1=1$ and $\phi_2=2$.}
	\label{fig:valfun}
\end{figure}

Figure \ref{fig:valfun} shows the value functions and the corresponding optimal investment levels at several different time points. In general, the agents will adopt the largest possible risk exposure when the portfolio value is negative due to risk-seeking over losses induced by the convex segment of the utility function. Investment level is the lowest when the portfolio value is at a small positive level. It is perhaps not too surprising because local risk-aversion is typically the highest for small positive wealth level. Meanwhile, the investment behaviours for larger positive wealth depend on the precise utility function of the agents. In the piecewise power (i.e. constant relative risk aversion alike) specification, investment level increases with wealth until it hits the delta limit implied by the dynamic risk constraint. For the piecewise exponential (i.e. constant absolute risk aversion alike) specification, the investment level will flat out at a constant level as wealth increases.

\begin{figure}[!htbp]
	\captionsetup[subfigure]{width=0.5\textwidth}
	\centering
	\subcaptionbox{S-shaped power utility function.}{\includegraphics[scale =0.515] {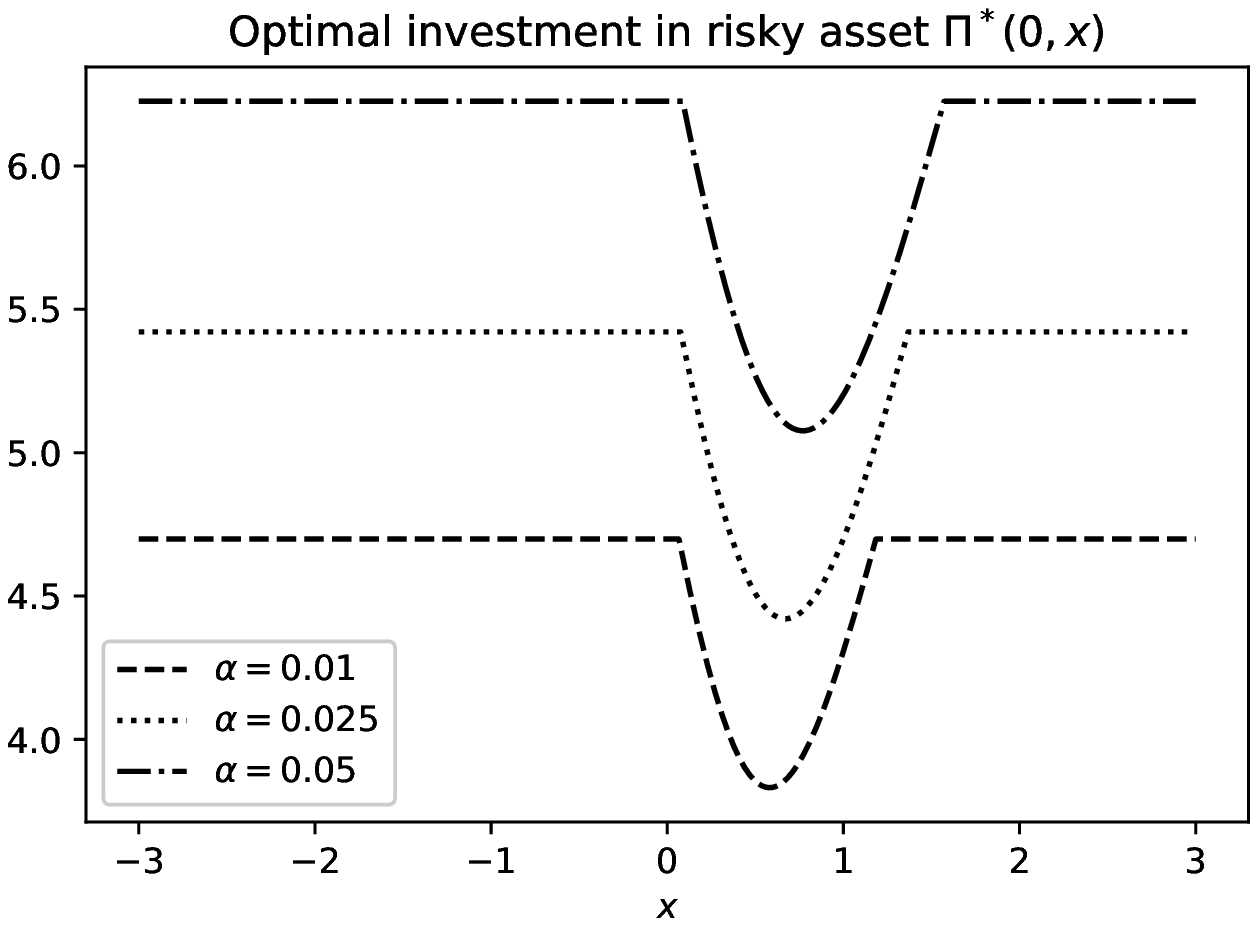}}
	\subcaptionbox{S-shaped exponential utility function.}{\includegraphics[scale =0.515]{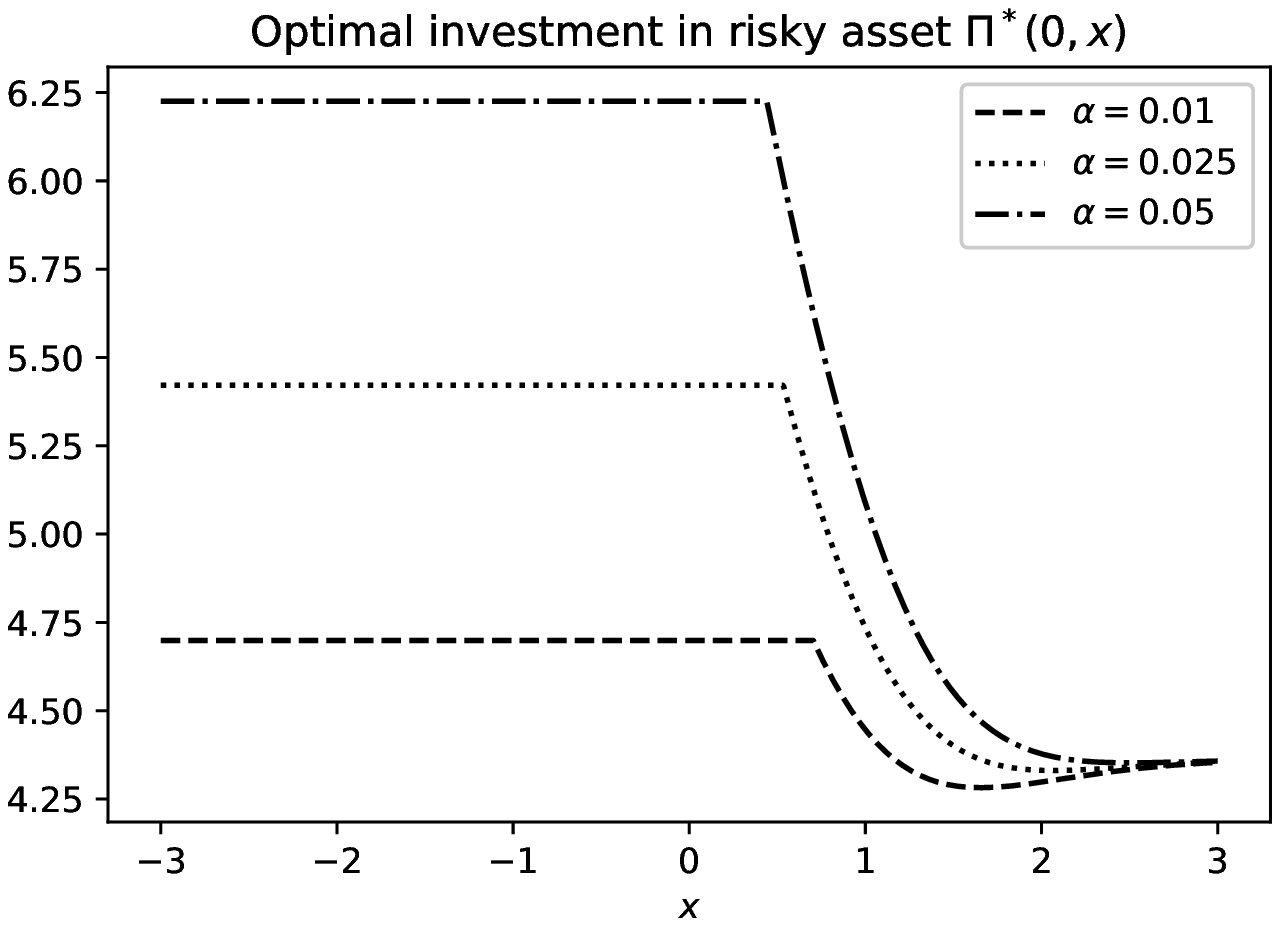}}
	\caption{Optimal investment level at $t=0$ under S-shaped power and exponential utility function for different values of Expected Shortfall confidence level $\alpha$. Based parameters used are $\mu=0.15$, $\sigma=0.25$, $r=0$, $T=1$, $\alpha=0.01$, $R=1$, $\Delta=30/250$, $\beta_1=\beta_2=0.5$, $k=2$, $\gamma_1=\gamma_2=0.55$, $\phi_1=1$ and $\phi_2=2$.}
	\label{fig:compstat}
\end{figure}

Figure \ref{fig:compstat} shows how the optimal investment level changes with the Expected Shortfall significance level. The results are intuitive: tighter the risk limit, more conservative the portfolio strategy.

We can measure in monetary terms the impact of a dynamic risk constraint on both the tail-risk-seeking trader and a risk averse manager who derives utility from the terminal value of the portfolio managed by the trader. Under a given set of model parameters, the maximal expected utility of the trader $V(t,x)$ and the optimal trading strategy $\Pi^*$ can be computed numerically. The certainty equivalent (CE) of the trader (with capital $x$ at time $t$) is defined as the value $C$ such that $V(t,x)=U(C)$. Economically, it is the fixed amount of wealth to be endowed by the trader to make him indifferent between this endowment and the opportunity to trade under a dynamic risk constraint. Likewise, the CE of the manager is defined as the value of $C$ solving $\mathbb{E}^{(t,x)}[U_m(X_T^{\Pi^*})]=U_m(C)$ where $U_m(\cdot)$ is the concave utility function of the manager.

As an example, consider a tail-risk-seeking trader with a unit of initial capital $x_0=1$ and his utility function has a piecewise power form. The risk averse manager has a utility function of $U_m(x)=-e^{-\eta x}$ and he imposes a dynamic ES constraint to risk-control the trader. Figure \ref{fig:ce} shows the time-zero CE of both the trader and the manager as a function of the risk limit $R$. 

\begin{figure}[!htbp]
	\centering
	\includegraphics[scale =0.6]{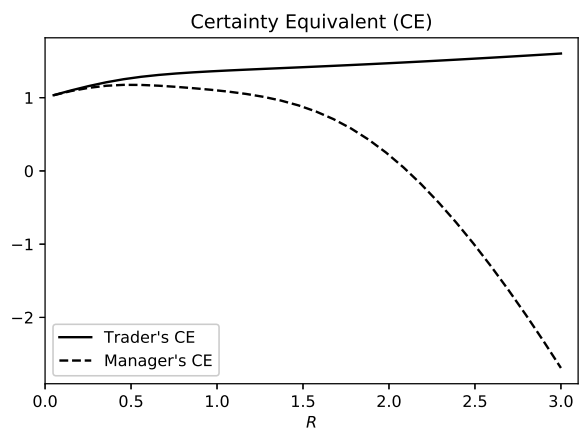}
	\caption{The time-zero certainty equivalent of the trader (with piecewise power utility function) and the risk averse manager (with exponential utility function) against the risk limit of a dynamic ES constraint. Parameters used are $\mu=0.15$, $\sigma=0.25$, $r=0$, $T=1$, $\alpha=0.01$, $\Delta=30/250$, $k=2$, $\alpha_1=\alpha_2=0.5$, $\eta=1$ and $x_0=1$.}
	\label{fig:ce}
\end{figure}

When $R$ is very close to zero, the CE of the trader and the risk manager are both around unity which is the initial trading capital. This is not surprising because under a very tight risk limit the trader essentially cannot purchase any risky asset. The portfolio under an admissible strategy is then almost riskless and the CE simply becomes the initial capital available (multiplied by the interest rate factor).

As $R$ increases, the CE of the trader gradually increases because a larger value of $R$ means the trader becomes less risk-constrained and therefore must be better off economically. On the other hand, the CE of the manager first increases slightly but then drops significantly. The CE of the manager improves at the beginning because a small but non-zero risk limit encourages the trader to invest conservatively in the risky asset which in turn creates value for the risk averse manager. However, when the risk limit is further relaxed, the trader takes more and more risk which starts becoming detrimental to the risk averse manager. Once $R$ goes above around 130\% of the initial capital, the CE of the manager goes below unity meaning that the trading activity now causes value destruction from the perspective of the manager. Indeed, when $R$ becomes arbitrarily large, Proposition \ref{prop:uncon} implies that the trader's CE will go to positive infinity while Theorem 5.4 of \cite{armstrong-brigo19} suggests the CE of the risk averse manager will become negative infinity. Figure \ref{fig:ce} highlights the conflict of interests between a tail-risk-seeking trader and a risk averse manager, and a slack risk management policy could easily result in drastic economic losses faced by the bank.

\section{A special case of zero excess return}
\label{sect:zerodrift}

Proposition \ref{prop:vis} provides a theoretical characterisation of the value function. However, it does not tell us much about the behaviours of the optimal portfolio strategy. In this section, we focus on a setup with $\mu=r=0$ (the assumption of $r=0$ is imposed for convenience only. The slightly more general case of $\mu=r$ can be handled by a change of numeraire technique.) The key idea is that in this special case we can exploit an equivalence between the risk-constrained portfolio optimisation problem and an optimal stopping problem. We show that the optimal trading strategy can be characterised in terms of a stopping time. As we will see soon, the state-space $[0,T]\times\mathbb{R}$ of the problem can be split into two regions: a trading region where the maximum possible amount is invested in the risky asset ($\Pi^*=k$ for some constant $k$) and a no-trade region where the agent opts to hold a pure cash position ($\Pi^*=0$).

As a preliminary discussion, investment motive vanishes in the case of $\mu=r=0$. Then whether the trader would participate in a fair gamble is purely driven by his risk appetite. Due to the S-shaped utility function, the trader is risk seeking over the domain of losses whereas he is risk averse over the domain of gains. Simple economic intuitions suggest that the trader prefers to gambling when the portfolio value is low, and prefers to taking all the risk off when the portfolio value is high. We therefore postulate that the optimal portfolio strategy has a ``bang-bang'' feature where the agent invests the maximum possible amount in the risky asset when the portfolio value is low. Once the portfolio value becomes sufficiently high, the trader's risk aversion dominates and he will immediately liquidate the entire holding in the risky asset. The postulated strategy can be stated in terms of a stopping time: the portfolio value evolves as a Brownian motion with maximum volatility (under the most aggressive admissible strategy) and stops when the agent decides to sell his entire risky asset holding and the portfolio value will remain unchanged thereafter. This inspires us to consider a simple optimal stopping problem introduced in the following subsection.

\subsection{An optimal stopping problem}

We introduce below an optimal stopping problem and verify some properties of its solution structure. Towards the end of this subsection, we will show that this optimal stopping problem and the risk-constrained portfolio optimisation problem \eqref{eq:valFun} are indeed equivalent. Before proceeding, we need to impose some slightly stronger assumptions on the utility function $U$ throughout this section.

\begin{assump}
The utility function $U:\mathbb{R}\to\mathbb{R}$ is a continuous, strictly increasing and strictly concave (resp. convex) $C^{2}$ function on $x>0$ (resp. $x<0$) with $U(0)=0$ and $\lim_{x\to +\infty} U'(x)=0$.
\label{assump:special}
\end{assump}

\begin{prop}
Suppose $X=(X_t)_{t\geq 0}$ has the dynamics of $dX_t=\nu dB_t$ where $\nu>0$ is a constant. Define an optimal stopping problem
\begin{align}
W(t,x):=\sup_{\tau\in\mathcal{T}_{t,T}}\mathbb{E}^{(t,x)}[U(X_{\tau})]
\label{eq:OSTvalfun}
\end{align}
where $\mathcal{T}_{t,T}$ is the set of $\mathcal{F}_t$-stopping times valued in $[t,T]$. The value function of problem \eqref{eq:OSTvalfun} is the unique viscosity solution to the HJB variational inequality
	\begin{align}
	\begin{cases}
	\min\left(-W_t-\frac{\nu^2}{2}W_{xx},W-U\right)=0,& t<T;\\
	W(T,x)=U(x),&t=T.
	\end{cases}
	\label{eq:hjbstop}
	\end{align}
	Define the continuation set $\mathcal{C}$ and the stopping set $\mathcal{S}$ as
	\begin{align}
	\mathcal{C}:=\{(t,x)\in [0,T)\times \mathbb{R}:W(t,x)>U(x)\},\quad \mathcal{S}:=\{(t,x)\in [0,T]\times \mathbb{R}:W(t,x)=U(x)\}.
	\label{eq:setCS}
	\end{align}
	The optimal stopping time is given by $\tau^*=\inf\{u\geq t: (u,X_u)\in \mathcal{S}\}$.
	
\end{prop}

\begin{proof}
	The relationship between the solution of an optimal stopping problem and the viscosity solution of the corresponding HJB variational inequality as well as the characterisation of the optimal stopping rule are standard - see for example \cite{oksendal-reikvam98}. Note that the techniques used in the proofs of Proposition \ref{prop:delta_suff} and \ref{prop:vis} can be adopted here to show that the value function $W(t,x)$ has at most a linear growth in $x$, which in turn confirms the uniqueness of the viscosity solution. 
\end{proof}

The below important result characterises the optimal stopping region in a more economically intuitive manner. In particular, the optimal stopping rule is a simple time-varying threshold strategy where the agent stops the process when its value is sufficiently high. 
\begin{prop}
	There exists a continuous and decreasing function $b:[0,T)\to (0,\infty)$ with $\lim_{t\uparrow T}b(t)=0$ such that the stopping set in \eqref{eq:setCS} admits a representation of
	\begin{align}
	\mathcal{S}=\{(t,x)\in [0,T)\times \mathbb{R}:x\geq b(t)\}\cup \left\{\{T\}\times\mathbb{R}\right\}.
	\label{eq:stopset}
	\end{align}
\label{prop:stopset}
\end{prop}

\begin{proof}
See the appendix.
\end{proof}

Finally, we verify the equivalence of the portfolio optimisation problem \eqref{eq:valFun} and the optimal stopping problem \eqref{eq:OSTvalfun} under $\mu=r=0$. 

\begin{prop}
	Suppose $\mu=r=0$. For $i\in\{\text{VaR},\text{ES}\}$, let $V_{i}$ be the value function of the portfolio optimisation problem \eqref{eq:valFun} under $K=K_i$. Then
	\begin{align}
	V_{i}(t,x)=W(t,x; k^{i})
	\label{eq:v_eq_w}
	\end{align}
	where 
	\begin{align}
	k^{\text{VaR}}=-\frac{R}{ \Phi^{-1}(\alpha)\sqrt{\Delta}}>0,\qquad k^{\text{ES}}=\frac{R\alpha}{ \phi(\Phi^{-1}(\alpha))\sqrt{\Delta}}>0,
	\label{eq:ki}
	\end{align}
	and $W(t,x;\nu)$ is the value function of the optimal stopping problem \eqref{eq:OSTvalfun} with diffusion constant $\nu$. Moreover, an optimal portfolio strategy is
	\begin{align}
	\Pi^*_t=\frac{k^i}{\sigma} 1_{(X_t<b(t))}
	\label{eq:canoptpi}
	\end{align}
	with $b(\cdot)$ being the optimal stopping boundary function of the stopping set introduced in \eqref{eq:stopset} associated with problem \eqref{eq:OSTvalfun} (under the diffusion parameter $\nu= k_i$).
	\label{prop:equi}
\end{prop}

\begin{proof}
	When $\mu=r=0$,  Lemma \ref{lem:setA} implies that the set $K_i$ simplifies to $K_i=[-\frac{k^i}{\sigma},\frac{k^i}{\sigma}]$ where the $k^i$'s are defined in \eqref{eq:ki}. Moreover, the Hamiltonian in \eqref{eq:hjb} becomes
	\begin{align*}
	H_i(p,M)=H_i(M)=\sup_{\pi\in[-\frac{k^i}{\sigma},\frac{k^i}{\sigma}]}\frac{\sigma^2}{2}M \pi^2=
	\begin{cases}
	\frac{ (k^i)^2}{2}M,& M\geq 0; \\
	0,& M< 0.
	\end{cases}
	\end{align*}
	To verify \eqref{eq:v_eq_w}, it is sufficient to show that $W(t,x;\nu)$, the solution to \eqref{eq:OSTvalfun}, is also a solution to \eqref{eq:hjb} under the choice of $\nu=k^i$. For $(t,x)\in\mathcal{C}$, we have $W_{xx}\geq 0$ by Lemma \ref{lem:propW} in the Appendix. Then $W_t+H(W_{xx})=W_{t}+\frac{ (k^i)^2}{2}W=W_{t}+\frac{\nu^2}{2}W=0$. For $(t,x)\in \mathcal{S}$, we have $W(t,x)=U(x)$ and $W_t+\frac{ (k^i)^2}{2}W_{xx}=W_t+\frac{\nu^2}{2}W_{xx}\leq 0$ which gives $W_t=0$ and $W_{xx}\leq 0$. Then $W_t+H(W_{xx})=0+0=0$. Hence $W$ solves \eqref{eq:hjb}.
	
	The candidate strategy $\Pi^*$ defined by \eqref{eq:canoptpi} is clearly in the admissible set $\mathcal{A}(K_i)$. To verify its optimality, one can compute $\mathbb{E}^{(t,x)}[U(X_T^{\Pi^*})]$ and show that it attains the same value as $W(t,x)$. But it is clear since
	\begin{align*}
	X^{\Pi}_T&=X_t+\int_t^T \sigma \Pi^*_{s} dB_s=X_t+\int_t^T  k^i 1_{(X_s<b(s))} dB_s \\
	&=X_t+\int_t^T  k^i 1_{(\tau^* > s)} dB_s\\
	&=X_t+\nu\int_t^{\tau^*} dB_s=X_t+\nu (B_{\tau^*}-B_t)
	\end{align*}
	where the portfolio process coincides (in distribution) with the optimally stopped process in problem \eqref{eq:OSTvalfun}.
\end{proof}

Figure \ref{fig:region} gives a stylised plot of the optimal portfolio strategy. 
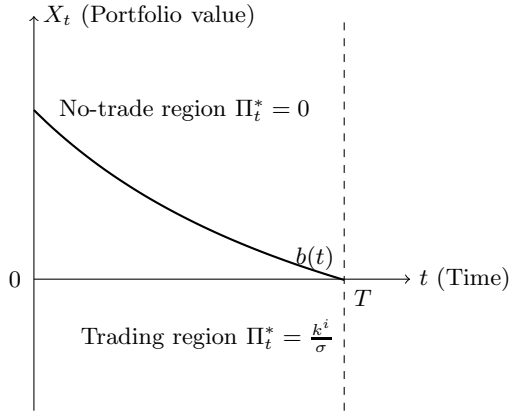
\begin{figure}[htbp]
	\center
	\begin{tikzpicture}[scale=5]
	\draw[->] (0,0.55) -- (1,0.55) node [right] {\small $t$ (Time)};
	\draw[->] (0,0.2) -- (0,1.25) node[right] {\small $X_{t}$ (Portfolio value)};
	\draw[dashed] (0.825,0.2) -- (0.825,1.25);
	\draw[thick][scale=1,domain=0:0.825,smooth,variable=\x] plot ({\x},{1/(\x+1}) node [above left] {\small $b(t)$};
	\node at (0.875,0.5) {\small $T$};
	\node at (-0.05,0.55) {\small $0$};
	\draw (0.1,0.4) node [right] {\small Trading region $\Pi^*_t=\frac{k^i}{\sigma}$};
	\node at (0.4,1) {\small No-trade region $\Pi^*_t=0$};
	\end{tikzpicture}
	\caption{A graphical illustration of the optimal portfolio strategy under the special case $\mu=r=0$. When the portfolio value is low, the agent invests the maximum possible amount in the risky asset by taking $\Pi_t^*=k^i/\sigma$, whereas when the portfolio value is high the agent takes all the risk off and sets $\Pi_t^*=0$. The critical boundary between risk-on and risk-off is given by a non-negative continuous, and decreasing function $b(t)$.}
	\label{fig:region}
\end{figure}

\subsection{Comparative statics}

Some comparative statics can be established to shed light on the policy implications of the dynamic risk constraint. We begin by offering a useful lemma.

\begin{lemma}
	Consider problem \eqref{eq:OSTvalfun} and let $b(t;\nu)$ be the optimal stopping boundary defined in \eqref{eq:stopset} under a fixed diffusion constant $\nu$. Then $b(t;\nu)$ is increasing in $\nu$.
	\label{lem:b_vol}
\end{lemma}

\begin{proof}
	Let $W^{(j)}$ be the value function and $b_j(t)$ be the corresponding optimal stopping boundary associated with problem \eqref{eq:OSTvalfun} under parameter $\nu_j$. Similarly, define
	\begin{align*}
	\mathcal{G}_j f:=-f_t-\frac{\nu_j^2}{2}f_{xx}.
	\end{align*}
	Fix $\nu_2>\nu_1$ and define $\bar{W}:=W^{(2)}-W^{(1)}$.
	
	On $x\geq b_1(t)$, we have $\bar{W}(t,x)=W^{(2)}(t,x)-W^{(1)}(t,x)=W^{(2)}(t,x)-U(x)\geq 0$. On $x<b_1(t)$, we have $\mathcal{G}_1 W^{(1)}=0$ and $\mathcal{G}_2 W^{(2)} \geq 0$ (note that the latter might need to be understood in a viscosity sense). Hence
	\begin{align*}
	0\leq \mathcal{G}_2 W^{(2)} - \mathcal{G}_1 W^{(1)} &= -W^{(2)}_t+W^{(1)}_t-\frac{\nu_2^2}{2}W^{(1)}_{xx}+\frac{\nu_1^2 }{2}W^{(1)}_{xx}\\
	&=-\frac{\partial}{\partial t}(W^{(2)}-W^{(1)})-\frac{\nu_2^2 }{2}\frac{\partial^2}{\partial x^2}(W^{(2)}-W^{(1)})-\frac{\nu_2^2-\nu_1^2}{2}W^{(1)}_{xx}\\
	&=\mathcal{G}_2 \bar{W}-\frac{\nu_2^2-\nu_1^2}{2}W^{(1)}_{xx}
	\end{align*}
	and therefore
	\begin{align*}
	\mathcal{G}_2 \bar{W}\geq \frac{\nu_2^2-\nu_1^2}{2}W^{(1)}_{xx}\geq 0
	\end{align*}
	for all $x<b_1(t)$ since $W_{xx}^{(1)}\geq 0$ on its continuation region by Lemma \ref{lem:propW} in the Appendix. Moreover, $\bar{W}(t,x)\geq 0$ on $x\geq b_1(t)$ and $\bar{W}(T,x)=U(x)-U(x)=0$, it follows from maximum principle that $W^{(2)}-W^{(1)}=\bar{W}\geq 0$ for all $x<b_1(t)$. Hence $W^{(2)}(t,x)\geq W^{(1)}(t,x)>U(x)$ on $x<b_1(t)$ from which we can conclude $b_2(t)\geq b_1(t)$, i.e. $b(t;\nu)$ is increasing in $\nu$.
\end{proof}

\begin{prop}
	In the special case of $\mu=r=0$, denote by $b(t;\theta)$ the trading boundary associated with the optimal strategy of the VaR/ES-constrained problem \eqref{eq:valFun} introduced in Proposition \ref{prop:equi} under a particular model parameter $\theta$. For $t$ being fixed, we have the following:
	\begin{enumerate}
		\item $b(t; T)$ is increasing in the trading horizon $T$;
		\item $b(t; R)$ is increasing in the risk limit level $R$;
		\item $b(t; \alpha)$ is increasing in the significance level of the VaR/ES measure $\alpha$;
		\item $b(t; \Delta)$ is decreasing in the risk evaluation window $\Delta$;
		\item $b(t; \sigma)$ does not depend on $\sigma$.
	\end{enumerate} 
\end{prop}

\begin{proof}
	From Proposition \ref{prop:equi}, the trading boundary of the optimal strategy is given by $b(t;k^i)$ which can be characterised by the optimal stopping boundary of problem \eqref{eq:OSTvalfun} with parameter $\nu=k^i$.
	
	Property 1 can simply be inferred from the fact that $b(t)$ is decreasing in $t$, and Property 2 to 5 immediately follow from Lemma \ref{lem:b_vol} by observing that both $k^{\text{VaR}}=-\frac{R}{\Phi^{-1}(\alpha)\sqrt{\Delta}}$ and $k^{\text{ES}}=\frac{R\alpha}{ \phi(\Phi^{-1}(\alpha))\sqrt{\Delta}}$ are increasing in $R$ and $\alpha$ (for $\alpha<0.5$), decreasing in $\Delta$, and does not depend on $\sigma$.
\end{proof}

Recall the optimal strategy is in form of $\Pi_t^*=\frac{k^i}{\sigma}1_{(X_t<b(t))}$. Hence $k^i$ governs the amount of investment given that the trader is in the trading region, and the location of $b(t)$ reflects how frequent the trader will be trading. Higher the value of $b(t)$, larger the regime of portfolio value under which the trader takes the most extreme risk exposure. Imposing dynamic risk constrains can curb such behaviours. As a first order effect, a strict risk limit (low $R$ or $\alpha$) reduces $k^{i}$ which limits the position value of risky asset investment. This restricts the volatility of the portfolio return and thus the trader might find it less attractive to gamble despite the non-concavity of his utility function. As a result, it also leads to a shrunk region of trading, i.e. $b(t)$ is lowered. Alternatively, the trader will trade less often if the trading horizon $T$ is reduced. This could potentially be achieved in practice by shortening the performance evaluation horizon.

\section{Derivatives trading under dynamic risk constraints}
\label{sect:derivatives}

In the context of portfolio optimisation under a complete market, it is typically not important to distinguish a ``delta-one'' trader (who is constrained to trade only in
the underlying stock and a risk-free account) and a derivatives trader (who can purchase any payoff structure contingent on the underlying stock price). It is because market completeness implies that perfect replication of any arbitrary claim is feasible and hence derivatives securities are redundant. This insight is exploited heavily to facilitate the martingale duality method where a dynamic portfolio selection problem is converted into a static problem of optimal payoff design.

Our main results in Section \ref{sect:mainresults} and \ref{sect:zerodrift} apply to a delta-one trader, in which case the expected shortfall of the portfolio is determined by the delta of the portfolio and this is a key ingredient in our calculation. 

However, the results will change drastically if the trader has access to the derivatives market. A trader with limited liability who is allowed to purchase arbitrary derivative
securities at the Black-Scholes price will be able to achieve arbitrarily high expected utilities under any expected shortfall constraint by pursuing a martingale type strategy. The essential idea is 
to use Theorem 4.1 of \cite{armstrong-brigo19} to find a derivative which comfortably meets the expected shortfall constraint and provides the desired utility. If at some future point the market moves so that the expected shortfall constraint hits the limit, then the trader may
apply Theorem 4.1 \cite{armstrong-brigo19} to find a new derivative which still yields the desired expected utility and which ensures that the constraint again comfortably met. It is possible to construct a strategy
so that with probability $1$, the trader will only need to rebalance their portfolio in this way a finite
number of times. We give a proof of this in Appendix \ref{appendix:martingale}.

Why is a dynamic risk constraint effective against a delta-one trader but not a derivatives trader? It is because the replication of large quantities of out-of-money digital options will involve trading a massive notional of the underling stock in the bad state of the world, which the delta-one trader understands ex-ante will not be feasible under a given dynamic risk constraint. In contrast, the feasibility of a derivative position only depends on the current statistical profile of the payoff. The derivatives trader can therefore exploit the blindspot of a risk measure to ensure the massive tail-risk is not detected. Finally, the possibility to roll-over a derivative position allows the risk constraint to be satisfied throughout the entire trading horizon.

One might ask what alternative types of risk limits would be effective against
such a trader. Expected utility constraints give one possible answer.
For example, one can choose a concave increasing function $U_m$ of the form
\[
U_m(x)=-(-x)^\gamma {\mathbbm 1}_{x\leq 0}
\]
for $\gamma \in(1, \infty)$
and require that at each time
\[
\mathbb{E}[U_m(X^{(t)}_{t+\Delta})]\geq R
\]
where $X^{(t)}_s$ is the time-$s$ value of the derivatives portfolio held
by the trader at time $t$ and $R \in (-\infty,0)$ is a chosen risk limit. To
see that such a constraint would be effective, first note that there would be a minimum wealth at time $T$ needed to achieve such a utility constraint. This would implies that the trading strategy must achieve a minimum expected $u_M$ at time $T$ and one may then apply Theorem 5.3 of
\cite{armstrong-brigo19}.

\section{Concluding remarks}
\label{sect:conclude}

While VaR and ES are widely adopted by practitioners, the impact of such risk constraints on traders' behaviours are not necessarily well understood. This paper addresses the negative result of \cite{armstrong-brigo19} that a static VaR/ES measure does not work at all on a tail-risk-seeking trader. Our key result highlights that dynamic monitoring of the trading positions is crucial. Continuous re-evaluation of portfolio exposure demands traders to respect a delta notional limit at all time. This alone is sufficient to discourage excessive risk taking during market distress which is naturally attractive to a tail-risk-seeking trader. 

However, the dangerous combination of tail-risk-seeking preference and derivatives trading can pose challenges to risk management. The possibility to rebalance a derivative position allows the trader to pursue a martingale strategy where the trading losses and risk limit breaches can be indefinitely deferred. As the possible alternatives to statistical measures like VaR or ES, utility-based risk measures or other scenario-based assessments such as stress testing might be the superior tools for risk managing derivatives traders. It will be of both theoretical and practical interests to further explore the desirable features of an effective risk control mechanism which performs well beyond delta-one trading.  

\bibliographystyle{apalike}
\bibliography{ref}

\appendix

\section{Proofs}
\label{app:proof}

We first provide some prior properties of the value function \eqref{eq:OSTvalfun} in the following lemma.

\begin{lemma}
	The value function \eqref{eq:OSTvalfun} has the following properties:
	\begin{enumerate}
		\item $W(t,x)$ is continuous in $t$ and $x$.
		\item $W(t,x)$ is decreasing in $t$ and is increasing in $x$;
		\item $W(t,x)\in C^{1,2}(\mathcal{C})$ with $W_{xx}(t,x)\geq 0$ for all $(t,x)\in \mathcal{C}$ where $\mathcal{C}$ is defined in \eqref{eq:setCS}.
	\end{enumerate}
	\label{lem:propW}
\end{lemma}

\begin{proof}
	Property 1 is due to the standard comparison principle of viscosity solution. Property 2 can be easily inferred from the structure of the optimal stopping problem. Here we will prove Property 3.
	
	Fix a bounded open domain $\mathcal{O}$ in $\mathcal{C}$ and consider a boundary value problem
	\begin{align}
	\begin{cases}
	\mathcal{G}f:=-f_t-\frac{\nu^2}{2}f_{xx}=0,& (t,x)\in\mathcal{O} \\
	f=W,&(t,x)\in\partial \mathcal{O}
	\end{cases}
	\label{eq:dir}
	\end{align}
	Since the operator $\mathcal{G}$ is linear, standard PDE theory suggests that there exists a unique smooth solution $f\in\mathcal{C}^{1,2}$ to \eqref{eq:dir} on $\mathcal{O}$. But this $f$ also solves \eqref{eq:hjbstop} on $\mathcal{O}$. By uniqueness of the viscosity solution, we deduce $W=f$ on $\mathcal{O}$ such that $W\in C^{1,2}(\mathcal{O})$. Finally, $\mathcal{C}$ is an open set and thus by the arbitrariness of $\mathcal{O}$ the smoothness property of $W$ can be extended to the entire $\mathcal{C}$.
	
	Thanks to the $C^{1,2}(\mathcal{O})$ property, \eqref{eq:hjbstop} can be interpreted in the classical sense such that on $\mathcal{C}$ we have $\frac{\nu^2}{2}W_{xx}=-W_t\geq 0$ as $W$ is decreasing in $t$.
\end{proof}

\begin{proof}[Proof of Proposition \ref{prop:stopset}]
	We first prove a preliminary result that $[0,T)\times (-\infty,0)\subseteq \mathcal{C}$, i.e. it is always suboptimal to stop on the negative regime before the terminal time. Suppose on contrary there exists $(t',x')$ with $0\leq t'< T$ and $x'<0$ such that $(t',x')\in \mathcal{S}$. Then $W(t',x')=U(x')$. Now consider an alternative stopping rule
	\begin{align*}
	\tau_{\epsilon}:=\inf\{s\geq t': X_s\notin (x'-\epsilon,x'+\epsilon)\}\wedge T
	\end{align*}
	for some $0<\epsilon<-x'$. Let $p_{\epsilon}:=\mathbb{P}(\tau_{\epsilon}<T)$. Then
	\begin{align*}
	W(t',x')\geq \mathbb{E}^{(t',x')}[U(X_{\tau_{\epsilon}})]&=\mathbb{P}(\tau_{\epsilon}<T,X_{\tau_{\epsilon}}=x'-\epsilon)U(x'-\epsilon)+\mathbb{P}(\tau_{\epsilon}<T,X_{\tau_{\epsilon}}=x'+\epsilon)U(x'+\epsilon)\\
	&\qquad+\mathbb{P}(\tau_{\epsilon}=T)\mathbb{E}[U(X_T)|\tau_{\epsilon}=T] \\
	&=\frac{1}{2}U(x'-\epsilon)+\frac{1}{2}U(x'+\epsilon)+(1-p_{\epsilon})\left\{\mathbb{E}[U(X_T)|\tau_{\epsilon}=T]-\frac{1}{2}U(x'-\epsilon)-\frac{1}{2}U(x'+\epsilon)\right\}\\
	&> U(x') +(1-p_{\epsilon})\left\{U(x'-\epsilon)-\frac{1}{2}U(x'-\epsilon)-\frac{1}{2}U(x'+\epsilon)\right\}.
	\end{align*}
	using the strict convexity of $U$ on $x<0$ and the martingale property of the process $X$. It is not hard to observe that $p_{\epsilon}\uparrow 1$ as $\epsilon\downarrow 0$. We immediately obtain the required contradiction $W(t',x')>U(x')=W(t',x')$.
	
	The rest of the proof goes as follows:
	
	\begin{enumerate}[label=(\roman*)]
		\item \textbf{Existence and non-negativity of $b$}: 
		
		We first show that $W(t,x)-U(x)$ is decreasing in $x$ over $x\geq 0$ and $t\in[0,T)$. Fix an arbitrary $\beta>0$ and define $F(t,x):=W(t,x+\beta)$. By the linear structure of the underlying Brownian motion, it can be easily seen that $F(t,x)=\sup_{\tau\in\mathcal{T}_{t,T}}\mathbb{E}[U(X_{\tau}+\beta)]$ and hence $F$ is the (unique) viscosity solution to
		\begin{align}
		\begin{cases}
		\min\left\{\mathcal{G}F,F-U(x+\beta)\right\}=0,& t<T;\\
		F(t,x)=U(x+\beta),& t=T,
		\end{cases}
		\label{eq:F}
		\end{align}
		where $\mathcal{G}f:=-f_t-\frac{\nu^2}{2}f_{xx}=0$. Let $G(t,x):=W(t,x)+U(x+\beta)-U(x)$. Then whenever $\mathcal{G}W=0$, we have $$\mathcal{G}G=\mathcal{G}W+\mathcal{G}U(x+\beta)-\mathcal{G}U(x)=\frac{\nu^2}{2}(U''(x)-U''(x+\beta))\geq0$$
		as $U$ is concave, and whenever $W(t,x)=U(x)$ we have $G(t,x)=U(x+\beta)$. Moreover, $G(T,x)=U(x+\beta)$. Hence $G$ is a supersolution to \eqref{eq:F}. By maximum principle (in a viscosity sense), we deduce $G\geq F$ leading to
		\begin{align}
		W(t,x)-U(x)\geq W(x+\beta)-U(x+\beta),
		\label{eq:decrease}
		\end{align}
		i.e. $W(t,x)-U(x)$ is decreasing.
		
		Now we show that that for each fixed $t\in[0,T)$ there always exists $x_0$ such that $W(t,x_0)=U(x_0)$. Such $x_0$, if exists, must be strictly positive since it is suboptimal to stop in the negative regime. Then together with the fact that $W(t,x)-U(x)$ is decreasing in $x$ on $x\geq 0$, we conclude there exists a unique $b(t)\in(0,\infty)$ such that $W(t,x)=U(x)\iff x\geq b(t)$. This will be sufficient to justify the existence of a positive boundary function $b$ which characterises the stopping set \eqref{eq:stopset}.
		
		To complete the proof, suppose on contrary that $W(t,x)>U(x)$ for all $x$. Then $\{t\}\times\mathbb{R}\in\mathcal{C}$ on which $W$ is a $C^{2}$ increasing convex function in $x$. $W(t,x)-U(x)$ being decreasing in $x$ now implies $W_x(t,x)\leq U'(x)$ for all $x$. In turn $\lim_{s\to\infty }W_{x}(t,s)\leq \lim_{s\to\infty}U'(s)=0$ and hence $W(t,x)$ must be a constant independent of $x$. But with $W(t,x)>U(x)$ this must imply $W(t,x)=\sup_s U(s)$ for all $(t,x)\in[0,T)\times \mathbb{R}$. This can easily shown to be false based on the same ideas used in the proof for Proposition \ref{prop:delta_suff}.

		\item \textbf{Monotonicity of $b$}: Consider $(t_1,x)\in \mathcal{S}$ such that $W(t_1,x)=U(x)$. Then for any $t_2>t_1$, we have $0\leq W(t_2,x)-U(x)\leq W(t_1,x)-U(x)=0$ since $W(t,x)$ is decreasing in $t$. Hence $W(t_2,x)=U(x)$ and $(t_2,x)\in \mathcal{S}$ as well such that $b(t)$ must be decreasing.
		
		\item \textbf{Continuity of $b$}: We begin by showing that $b$ is right-continuous. Fix $t<T$ and consider a decreasing sequence $(t_n)_{n\geq 1}$ with $t_n\downarrow t$. Then for each $n$ we have $(t_n,b(t_n))\in \mathcal{S}$. Since the set $\mathcal{S}$ is closed, we have $(t,b(t+))\in \mathcal{S}$ as well such that $b(t+)\geq b(t)$. But $b(t+)\leq b(t)$ as $b$ is decreasing. We hence conclude $b(t)=b(t+)$.
		
		Now we show that $b(t)$ is left-continuous. Suppose on contrary that there exists $t<T$ such that $b(t-)>b(t)$. Define $\xi:=\frac{b(t-)+b(t)}{2}$ such that $0\leq b(t)<\xi<b(t-)$. Choose $s\in(0,t)$ and then $0\leq b(t)<b(t-)\leq b(s)$.
		By definition of $b(\cdot)$ and the smooth-pasting property, we have $W(s,b(s))=U(b(s))$ and $W_{x}(s,b(s))=U'(b(s))$.\footnote{Smooth-pasting must hold at $b(s)$ because $b(s)>0$ and $U'(x)$ exists for any $x>0$. See \cite{peskir-shiryaev06}.} Then
		\begin{align*}
		W(s,\xi)-U(\xi)=\int_{b(s)}^{\xi}(W_x(s,y)-U'(y))dy&=\int_{b(s)}^{\xi}\int_{b(s)}^{y}(W_{xx}(s,z)-U''(z))dzdy \\
		&=\int_{\xi}^{b(s)}\int_{y}^{b(s)}(W_{xx}(s,z)-U''(z))dzdy\\
		&\geq \int_{\xi}^{b(s)}\int_{y}^{b(s)} C dz dy=\frac{C}{2}(\xi-b(s))^2
		\end{align*}
		for some constant $C>0$ independent of $s$ where we have used Lemma \ref{lem:propW} that $W_{xx}\geq 0$ on $\mathcal{C}$ and $U$ is a strictly concave $C^2$ function on the positive domain. Since $W(s,x)$ is continuous in $s$, if we let $s\uparrow t$ we deduce $$W(t,\xi)-U(\xi)\geq \frac{C}{2}(\xi-b(t-))^2>0.$$ But $b(t)<\xi$ and hence $(t,\xi)\in\mathcal{S}$ which implies $W(t,\xi)=U(\xi)$. We arrive at the required contradiction.
		
		\item \textbf{Limiting behaviour of $b$}: Suppose $b(T):=\lim_{t\uparrow T} b(t)>0$. Then let $\xi:=\frac{b(T)}{2}>0$ and with the same argument as in part (iii) of the proof we can deduce $W(s,\xi)-U(\xi)>\frac{C}{2}(\xi-b(s))^2$ for some constant $C>0$ and $s<T$. Contradiction can be obtained again by letting $s\uparrow T$ on recalling the terminal condition that $W(T,x)=U(x)$ for all $x$.
	\end{enumerate}
\end{proof}

\section{Ineffectiveness of dynamic Expected Shortfall constraint on a derivative trader}
\label{appendix:martingale}

We will assume that $U(x)=0$ for all negative $x$, so we are
specializing to the case of an investor with limited liability.
We will consider a Black--Scholes market with $\mu > r$.

Consider a derivatives trader who has a given budget and who wishes
to purchase a sequence of European options all with maturity $T$
to maximize their expected utility. We suppose that this trader
is subject only to a self-financing constraint and that at all times $t\in[0,T)$
they must meet an expected shortfall constraint at confidence level $\alpha$
with time interval $T-t$. We will show that such a trader can achieve an expected utility greater than or equal to $u$ for any $u<\sup U$.

Let $f:(0,\infty)\to \R$ be the payoff function of a European derivative with maturity $T$. We
will write ${\mathcal P}(f,S,t)$, ${\mathcal U}(f,S,t)$ and $\ES_\alpha(f,S,t)$ for the price, expected $U$ utility and expected shortfall of this derivative at time $t$ given that $S_t=S$ (the time horizon for the expected shortfall calculation being the remaining time to maturity, $T-t$).

Given constants $k,h,\ell \in \R$ with $\ell < 0 < h$ and
we define a
digital payoff function $f^{k,h,\ell}$ by
\[
f^{k,h,\ell}(S_T)=\begin{cases}
h & S_T \geq e^k \\
\ell & \text{otherwise}. 
\end{cases}
\]
We will write ${\mathcal D}$ for the set of all such digital options.

The proof of Theorem 4.1 in \cite{armstrong-brigo19}
shows that we may find $f^{k,h,\ell}\in {\mathcal D}$ satisfying
\begin{equation}
\begin{split}
{\mathcal P}(f^{k,h,\ell},S,t) &\leq -1 \\
\ES_\alpha(f^{k,h,\ell},S,t) &\leq -1 \\
{\mathcal U}(f^{k,h,\ell},S,t)&\geq u
\label{eq:formalResult1}
\end{split}
\end{equation}
for any $u<\sup U$. Moreover, we may take $k$ to be arbitrarily small and
one may also require $|\ell|>h$. We prove an extension of this result in the lemma below.

\begin{lemma}
Given any constants $\lambda_1,\lambda_2$ with $\lambda_1>0$ and $\lambda_2\in (0,\mu-r)$, there exists  $f^{k,h,\ell}\in {\mathcal D}$ such that
\begin{equation}
\begin{split}
{\mathcal P}(f^{k,h,\ell},Se^{\lambda_1 (T-t)},t) &\geq 1 \\
{\mathcal P}(f^{k,h,\ell},S,t) &= -1 \\
\ES_\alpha(f^{k,h,\ell},Se^{-\lambda_2 (T-t)},t) &\leq -1\\
{\mathcal U}(f^{k,h,\ell},S,t)&\geq u
\label{eq:newResult1}
\end{split}
\end{equation}
for any $u<\sup U$. The $k$ can be arbitrarily small.
\label{lemma:digital}
\end{lemma}

\begin{proof}
Calculating the price, expected shortfall and expected utility explicitly one finds
that
\begin{equation}
\begin{split}
{\mathcal P}(f^{k,h,\ell},S,t)&=e^{-r(T-t)}\left( \ell \Phi\left( \frac{k-\ln S - \nu(T-t)}{\sigma \sqrt{T-t}} \right)
+ h \left(1 - \Phi\left( \frac{k - \ln S - \nu(T-t)}{\sigma\sqrt{T-t}} \right) \right)
\right) \\
\ES_\alpha(f^{k,h,\ell},S,t)&=-\frac{1}{\alpha}\left( \ell \Phi\left( \frac{k -\ln S - \tilde{\nu}(T-t)}{\sigma\sqrt{T-t}}
\right) + h \left(\alpha - \Phi\left( \frac{k - \ln S - \tilde{\nu}(T-t)}{\sigma\sqrt{T-t}}\right) \right)
\right) \\
{\mathcal U}(f^{k,h,\ell},S,t)&=U(h)\left(1 - \Phi\left( \frac{k - \ln S - \tilde{\nu}(T-t)}{\sigma\sqrt{T-t}}\right) \right)
\end{split}
\label{eq:formalResult2}
\end{equation}
for $k<\ln S+\tilde{\nu}(T-t)+\sigma\sqrt{T-t}\Phi^{-1}(\alpha)$, where $\nu :=  r-\frac{1}{2}\sigma^2$ and $\tilde{\nu} := \mu-\frac{1}{2}\sigma^2$. Fix $h>0$. The budget constraint ${\mathcal P}(f^{k,h,\ell},S,t) = -1$ allows us to express $\ell$ in terms of $k$ and $h$ via
\begin{align*}
\ell=h-\frac{h+e^{r(T-t)}}{\Phi\left( \frac{k-\ln S - \nu(T-t) }{\sigma\sqrt{T-t}} \right)}.
\end{align*}
Then ${\mathcal P}(f^{k,h,\ell},Se^{\lambda_1 (T-t)},t)$ and $\ES_\alpha(f^{k,h,\ell},Se^{-\lambda_2 (T-t)},t)$ can be written as
\begin{align*}
{\mathcal P}(f^{k,h,\ell},Se^{\lambda_1 (T-t)},t)&=e^{-r(T-t)}\left(h-(h+e^{r(T-t)})\frac{\Phi\left( \frac{k-\ln S-\lambda_1(T-t) - \nu(T-t)}{\sigma \sqrt{T-t}} \right)}{\Phi\left( \frac{k-\ln S - \nu(T-t)}{\sigma \sqrt{T-t}} \right)}\right),\\ \ES_\alpha(f^{k,h,\ell},Se^{-\lambda_2 (T-t)},t)&=-h+\frac{h+e^{r(T-t)}}{\alpha}\frac{\Phi\left( \frac{k-\ln S+\lambda_2(T-t) - \tilde{\nu}(T-t)}{\sigma\sqrt{T-t}} \right)}{\Phi\left( \frac{k-\ln S - \nu(T-t)}{\sigma\sqrt{T-t}} \right)}.
\end{align*}
For as long as $\lambda_1>0$ and $\lambda_2\in (0,\mu-r)$, we have
\begin{align*}
\lim_{k\to -\infty}{\mathcal P}(f^{k,h,\ell},Se^{\lambda_1 (T-t)},t)=e^{-r(T-t)}h,\qquad \lim_{k\to -\infty}\ES_\alpha(f^{k,h,\ell},Se^{-\lambda_2 (T-t)},t)=-h,\qquad \lim_{k\to-\infty} {\mathcal U}(f^{k,h,\ell},S,t)=U(h).
\end{align*}
The result immediately follows since $h$ is arbitrary.

\end{proof}

Let us write
\[
d_{\nu}(S,k,\xi):=\sigma^{-1}(\xi^{-1}( k - \ln S) -\nu \xi).
\]
This will have a positive partial derivative with respect to $\xi$ whenever
\[
k < \ln S - \nu \xi^2.
\]
Hence if choose $k$ such that
\[
k < \ln S - |\nu| T=\ln S - |\mu - \tfrac{1}{2} \sigma^2| T
\]
then $d_{\nu}(S,k,\xi)$ will be an increasing function of $\xi$ in the range
$\xi \in (0,\sqrt{T}].$ Since
\[
ES_{\alpha}(f^{k,h,\ell},S,t)=-\frac{1}{\alpha}\left[
\ell \Phi( d_{\nu}(S,k,\sqrt{T-t}) ) + h(\alpha-\Phi(d_{\nu}(S,k,\sqrt{T-t})))\right]
\]
we see that for sufficiently small $k$, when $|\ell|>|h|$, $ES_{\alpha}(f^{k,j,\ell},S,t)$ will be decreasing in $t$
for $t\in[0,T)$ (consider $S$ and $f^{k,h,\ell}$ as fixed). Similarly, we can deduce ${\mathcal P}(f^{k,h,\ell},S,t)$ is increasing in $t$ for $t\in[0,T)$ for sufficiently small $k$. It now follows from Lemma \ref{lemma:digital} that given $S_t=S \in \R_{\geq 0}$ and $t \in [0,T)$ we may find $f^{k,h,\ell} \in {\mathcal D}$ satisfying 
\begin{equation}
\begin{split}
{\mathcal P}(f^{k,h,\ell},Se^{\lambda_1 (T-t)},s) &\geq 1 \\
{\mathcal P}(f^{k,h,\ell},S,t) &\leq -1 \\
\ES_\alpha(f^{k,h,\ell},S e^{-\lambda_2 (T-t)},s) &\leq -1 \\
{\mathcal U}(f^{k,h,\ell},S,t)&\geq  u.
\end{split}
\label{eq:formalResult4}
\end{equation}
for all $s \in [t,T)$, $\lambda_1>0$, $\lambda_2\in(0,\mu-r)$ and $u<\sup U$.

Suppose that at a given time $t$ the stock price is $S_t$.
By purchasing an arbitrarily large quantity of the option with payoff $f^{k,h,\ell}$ satisfying \eqref{eq:formalResult4} we can meet any cost or expected shortfall constraint at time $t$. By choosing $k$ sufficiently small we may
ensure that the probability this option has a negative payoff is as small as we like. Thus we may find an option that meets
our current budget, meets a given expected shortfall and ensures
that the expected utility for the trader is greater than or equal to $u<\sup U$.
Furthermore, when the stock price rises to $S_t e^{\lambda_1(T-t)}$ the option position can be liquidated for an arbitrarily high positive value. On the other hand, this option will continue to meet the expected shortfall constraint
until the stock price falls to $S_t e^{-\lambda_2(T-t)}$. When this occurs, the trader can opt to rebalance the option position. Let $\pi$ be the probability that a rebalancing occurs, which is the probability that the stock price level visits $S_t e^{-\lambda_2(T-t)}$ before $S_t e^{\lambda_1(T-t)}$ in the time interval $[t,T]$. By the scaling properties of geometric Brownian motion, one can show that if $\mu\geq \sigma^2/2$ then $\pi$ is bounded above by the constant $\lambda_1/(\lambda_1+\lambda_2)$.\footnote{For the case of $\mu<\sigma^2/2$, an upper bound can be derived by using the fact that the probability of a drifting Brownian motion $X_t:=\eta t + \sigma B_t$ hitting level $-a$ before $b$ over an infinite horizon is $\frac{1-\exp(-\frac{2\eta}{\sigma^2}b)}{\exp(\frac{2\eta}{\sigma^2}a)-\exp(-\frac{2\eta}{\sigma^2}b)}$.}

We define a sequence of stopping times $(t_i)_{i \in \N}$ inductively as follows. We define $t_0=0$ and construct $f_0 \in {\mathcal D}$ such that equation \eqref{eq:formalResult4} holds. Let $t_1$ be the smaller of $T$ and the first time $t^\prime>t_0$ satisfying $S_{t^\prime} = S_{t_0} e^{-\lambda_2 (T-t_0)}$ or $S_{t^\prime} = S_{t_0} e^{\lambda_1 (T-t_0)}$. If $S_{t_1}= S_{t_0} e^{\lambda_1 (T-t_0)}$, the option is liquidated at a positive value. Then the trader deposits the proceed in the riskfree account and no further action is taken until the maturity time $t=T$ (equivalent to purchasing a claim with constant payoff equal to the future value of the trader's current wealth). Else if $S_{t_1}= S_{t_0} e^{-\lambda_2 (T-t_0)}$, the trader roll into a new position $f_1\in {\mathcal D}$ satisfying equation \eqref{eq:formalResult4} using the current market value of $f_0$ at $t=t_1$ as the new initial budget.

Once $t_i$ has been defined and the position is not yet liquidated, we choose $f_i \in {\mathcal D}$ such that equation \eqref{eq:formalResult4} holds again
for $S=S_{t_i}$ and $t=t_i$. We define $t_{i+1}$ to be the smaller of $T$ and the first time $t^\prime>t_{i}$ satisfying $S_{t^\prime} = S_{t_i} e^{-\lambda_2 (T-t_i)}$ or $S_{t^\prime} = S_{t_i} e^{\lambda_1 (T-t_i)}$. At each time $t_i<T$ where the option is not liquidated, the trader
may rebalance by purchasing $\alpha>0$ units of a derivative with payoff $f_i$ to guarantee
that $u_i:={\mathcal U}(\alpha f_i,S_t,t)\geq u$, to meet the budget constraint
and to ensure that expected shortfall constraint will hold until 
the stopping time $t_{i+1}$. If the option is liquidated at $t=t_i$, the resulting proceed is arbitrarily large by construction of $f_{i-1}$ and hence depositing this amount of wealth in the riskfree account can deliver any (riskfree) utility level $u$ at maturity.

The probability that the trader needs to rebalance the portfolio $n$ or more times is $\pi^n$. So with probability $1$, the trader will only need to rebalance finitely often. The expected utility conditioned on the stock price at time $t$ will be an increasing function of $S_t$. Hence the expected utility conditioned on rebalancing the portfolio exactly $n$ times will be greater than or equal to $u_n$, since rebalancing only occurs if the stock price drops. Hence the overall expected utility of the strategy will be greater than or equal to $u$. Our assumption of limited liability ensures that our trader is unconcerned by events of probability $0$, so it
is acceptable to assign an expected utility to this strategy even though
it is a martingale strategy.

This shows that a derivatives trader subject only to an expected
shortfall constraint with the expected shortfall always calculated at time $T$
is able to achieve any expected utility less than $\sup U$. We now consider
instead a trader who invests on a time interval $[0,T]$ subject to
to expected shortfall constraints with a fixed time interval $\Delta\leq T$ used
to calculate the expected shortfall. This trader may choose any acceptable strategy until $T-\Delta$. Over the time interval $[T-\Delta,T)$ they may pursue an essentially identical strategy to the one given above except
using derivatives whose payoff at time $t+\Delta$ is given by a
function $f^{k,h,\ell}(S_T)$ for some $f^{k,h,\ell} \in {\mathcal D}$.
Hence such a trader can also achieve a utility greater than or equal to $u$.

\section{Extension to multiple risky assets}
\label{app:multi}

We discuss how several key results in this paper will change when there are multiple risky asset. Suppose there are $n$ risky assets in the Black-Scholes economy and the price process vector $S_t:=(S_t^{(1)},S_t^{(1)},\dots,S_t^{(n)})'$ has the dynamics
\begin{align*}
dS_t=Diag(S_t) \mu dt + \sigma dB_t
\end{align*}
where $Diag(S_t):=Diag(S_t^{(1)},S_t^{(2)},\dots,S_t^{(n)})$ is a $n\times n$ diagonal matrix, $\mu$ is a $n\times 1$ constant vector, $\sigma$ is a $n\times m$ constant matrix and $B$ is a $m$-dimensional Brownian motion. The portfolio strategy $\Pi$ is now valued in $\mathbb{R}^n$ and the portfolio value process $X$ has dynamics of
\begin{align}
dX_t=[rX_t+\Pi_t'(\mu-r e)]dt + \Pi_t' \sigma dB_t
\label{eq:multiassetwealth}
\end{align}
with $e$ being a vector of unity. Following a similar derivation as in Section \ref{sect:setup}, if the risk manager assumes the risky assets holding $\Pi$ is held fixed over the risk evaluation horizon then the projected portfolio loss on $[t,t+\Delta]$ is a normal random variable $L_t$ with
\begin{align*}
\mathbb{E}[L_t]=-\frac{e^{r\Delta}-1}{r}\Pi_t'(\mu-re),\quad \text{Var}(L_t)=\frac{e^{2r\Delta}-1}{2r}\Pi_t'\Sigma \Pi_t
\end{align*}
where $\Sigma:=\sigma\sigma'$ is the variance-covariance matrix of the risky assets return. 

It is now straightforward to write down the new admissible sets under the VaR and ES constraint as
\begin{align}
K_{\text{VaR}}:=\left\{\pi\in\mathbb{R}^n: -\frac{e^{r\Delta}-1}{r}\pi'(\mu-re)-\Phi^{-1}(\alpha) \sqrt{\frac{e^{2r\Delta}-1}{2r}\pi'\Sigma \pi}\leq R\right\}
\label{eq:setKVaR2}
\end{align} 
and
\begin{align}
K_{\text{ES}}:=\left\{\pi\in\mathbb{R}^n: -\frac{e^{r\Delta}-1}{r}\pi'(\mu-re)+\frac{\phi(\Phi^{-1}(\alpha))}{\alpha} \sqrt{\frac{e^{2r\Delta}-1}{2r}\pi'\Sigma \pi}\leq R\right\}
\label{eq:setKES2}
\end{align} 
respectively. The portfolio optimisation problem is to solve
\begin{align}
V(t,x):=\sup_{\Pi\in\mathcal{A}(K)}\mathbb{E}^{(t,x)}[U(X_T)]
\label{eq:valfun2}
\end{align}
where $\mathcal{A}(K):=\{\Pi:\int_0^T |\Pi_t|^2 dt<\infty \text{ } \mathbb{P}-a.s,\quad \Pi(t,\omega)\in K \quad \mathcal{L}\otimes \mathbb{P}\text{-a.e. } (t,\omega)\}$ subject to the dynamics \eqref{eq:multiassetwealth}. A dynamic VaR and ES constraint can be incorporated by the choice of $K=K_{\text{VaR}}$ and $K_{\text{ES}}$.

Now we show that Proposition \ref{prop:delta_suff} can be extended to the multi-asset setup.
\begin{prop}
	For the optimisation problem \eqref{eq:valfun2}, if the set $K$ is bounded then for every $t<T$ there exists $x$ such that $V(t,x)<\sup_{s}U(s)$
	\label{prop:multi_suff}
\end{prop}

\begin{proof}
	Based on the same ideas in the proof of Proposition \ref{prop:delta_suff}, for any S-shaped $U$ there exists some $m>0$ and $C>0$ such that $U(x)\leq C+m x^{+}=:\bar{U}(x)$. For as long as $K$ is bounded, there exists $b>0$ such that
	$$K\subseteq \left\{(\pi_1,\pi_2,\dots,\pi_n)'\in\mathbb{R}^n: -b\leq \pi_i\leq b \quad\forall i=1,2,\dots,n\right\}=:\bar{K}.$$ Let $\bar{\mu}:=b\sum_{i=1}^{n}|\mu_i-r|$. Given $X_t=x$, for any $\Pi\in\mathcal{A}(\bar{K})$ we have
	\begin{align*}
	X^{\Pi}_T&=xe^{r(T-t)}+\int_t^T e^{r(T-s)}\Pi_t' (\mu-re) ds + \int_t^T e^{r(T-s)} \Pi_s' \sigma dB_s \\
	&\leq xe^{r(T-t)}+\bar{\mu}\int_t^T e^{r(T-s)} ds + \int_t^T e^{r(T-s)} \Pi_s' \sigma dB_s=:\bar{X}^{\Pi}_T
	\end{align*}
	almost surely. Hence
	\begin{align}
	V(t,x)&\leq \sup_{\Pi_t\in \mathcal{A}(\bar{K})} \mathbb{E}^{(t,x)}[\bar{U}(X^{\Pi}_{T})]\nonumber \\
	&\leq \sup_{\Pi\in \mathcal{A}(\bar{K})} \mathbb{E}^{(t,x)}[\bar{U}(\bar{X}^{\Pi}_{T})]=:\bar{V}(t,x)
	\label{eq:relaxcontrol}
	\end{align}
	where $\bar{X}$ has dynamics of
	\begin{align*}
	d\bar{X}_t=(r\bar{X}_t +\bar{\mu} )dt+ \Pi_t' \sigma dB_t,\qquad \bar{X}_t=x.
	\end{align*}
	Since $\bar{U}$ is convex, based on the same augment in the proof of Proposition \ref{prop:delta_suff} we expect the optimal control for problem \eqref{eq:relaxcontrol} is a constant process given by
	\begin{align*}
	\Pi^*_t=\sup_{\pi:\pi_{i}\in\{-K,K\}}\pi'\Sigma\pi
	\end{align*}
	for all $t$. i.e. one should make the volatility of the process to be as large as possible. Moreover, the optimally controlled process $\bar{X}^{\Pi^*}$ is simply a drifting Brownian motion. Hence we can show that the corresponding value function $\bar{V}$ has the same form as in \eqref{eq:upbound2} and in turn the same result follows immediately.
\end{proof}

In the single-asset case, whether the set $K_i$ is bounded is solely determined by the Sharpe ratio of the asset relative to the risk management parameter $M_i$ as shown in Lemma \ref{lem:setA}. We illustrate that a similar sufficient condition holds in the case with two risky assets.

\begin{lemma}
	Consider a two-asset economy where the mean return vector is $\mu=(\mu_1,\mu_2)'$ and the variance-covariance matrix is $\Sigma=(\sigma_1^2,\sigma_1\sigma_2\rho;\sigma_1\sigma_2\rho,\sigma_2^2)$. Recall the definitions of $M_i$ for $i\in\{\text{VaR},\text{ES}\}$ as per \eqref{eq:Mconst}.
	If
	\begin{align}
	\sqrt{\frac{\left(\frac{\mu_1-r}{\sigma_1}\right)^2+\left(\frac{\mu_2-r}{\sigma_2}\right)^2-2\left(\frac{\mu_1-r}{\sigma_1}\right)\left(\frac{\mu_2-r}{\sigma_2}\right)\rho}{1-\rho^2}}<M_i,
	\label{eq:multisharpe}
	\end{align}
	then the set $K_i$ defined by \eqref{eq:setKVaR2} or \eqref{eq:setKES2} is bounded.
	\label{lem:twoasset}
\end{lemma}

\begin{proof}
	We prove the result for the case of ES constraint $i=\text{ES}$. The result under VaR constraint can be obtained similarly.
	
	If $(\pi_1,\pi_2)'\in K_{\text{ES}}$ then
	\begin{align*}
	&-\frac{e^{r\Delta}-1}{r}[(\mu_1-r)\pi_1+(\mu_2-r)\pi_2]+\frac{\phi(\Phi^{-1}(\alpha))}{\alpha} \sqrt{\frac{(e^{2r\Delta}-1)}{2r}(\sigma_1^2\pi_1^2+\sigma_2^2\pi_2^2+2\sigma_1\sigma_2\rho\pi_1\pi_2)}\leq R \\
	&\iff \frac{\phi(\Phi^{-1}(\alpha))}{\alpha} \sqrt{\frac{(e^{2r\Delta}-1)}{2r}(\sigma_1^2\pi_1^2+\sigma_2^2\pi_2^2+2\sigma_1\sigma_2\rho\pi_1\pi_2)}\leq R+\frac{e^{r\Delta}-1}{r}[(\mu_1-r)\pi_1+(\mu_2-r)\pi_2] \\
	&\implies \left[\frac{\phi(\Phi^{-1}(\alpha))}{\alpha} \sqrt{\frac{(e^{2r\Delta}-1)}{2r}(\sigma_1^2\pi_1^2+\sigma_2^2\pi_2^2+2\sigma_1\sigma_2\rho\pi_1\pi_2)}\right]^2\leq \left[R+\frac{e^{r\Delta}-1}{r}[(\mu_1-r)\pi_1+(\mu_2-r)\pi_2]\right]^2.
	\end{align*}
	On rearranging, the last inequality becomes
	\begin{align*}
	f(\pi_1,\pi_2):=A\pi_1^2 + B\pi_1\pi_2+C\pi_2^2+D\pi_1+E \pi_2 \leq R^2
	\end{align*}
	where
	\begin{align*}
	A&:=\left(\frac{\phi(\Phi^{-1}(\alpha))}{\alpha}\right)^2\left(\frac{e^{2r\Delta}-1}{2r}\right)^2\sigma_1^2-2\left(\frac{e^{r\Delta}-1}{r}\right)^2 (\mu_1-r)^2,\\
	B&:=2\left(\frac{\phi(\Phi^{-1}(\alpha))}{\alpha}\right)^2 \left(\frac{e^{2r\Delta}-1}{2r}\right)\sigma_1\sigma_2\rho-2\left(\frac{e^{r\Delta}-1}{r}\right)^2(\mu_1-r)(\mu_2-r),\\
	C&:=\left(\frac{\phi(\Phi^{-1}(\alpha))}{\alpha}\right)^2\left(\frac{e^{2r\Delta}-1}{2r}\right)^2\sigma_2^2-\left(\frac{e^{r\Delta}-1}{r}\right)^2 (\mu_2-r)^2,\\
	D&:=-2R\left(\frac{e^{r\Delta}-1}{r}\right)(\mu_1-r),\\
	E&:=-2R\left(\frac{e^{r\Delta}-1}{r}\right)(\mu_2-r).
	\end{align*}
	The set $\tilde{K}:=\{(\pi_1,\pi_2)'\in\mathbb{R}^2:f(\pi_1,\pi_2) \leq R^2\}$ is bounded if and only if the conic section $f(\pi_1,\pi_2)=R^2$ is an ellipse, or equivalently $B^2-4AC<0$. This condition can be explicitly written in terms of the model parameters as in \eqref{eq:multisharpe}. The result immediately follows on noticing that $K_{\text{ES}}$ is a subset of $\tilde{K}$.
\end{proof}

\begin{cor}
	A dynamic risk constraint $i\in\{\text{VaR},\text{ES}\}$ is effective in a two-asset economy if condition \eqref{eq:multisharpe} holds.
	\label{cor:multieffective}
\end{cor}
\begin{proof}
This immediately follows from Lemma \ref{lem:twoasset} and Proposition \ref{prop:multi_suff}.
\end{proof}
Comparing the results in Lemma \ref{lem:twoasset} to the single-asset case in Lemma \ref{lem:setA}, we can see that a dynamic risk constraint will translate into a bound on the trading strategy provided that the same risk management parameter $M_i$ is sufficiently strict relative to the quality of the assets. But the criteria of assets quality will now take the correlation $\rho$ into account to reflect the benefits of diversification. 



\end{document}